\pdfoutput=1 
\documentclass[twocolumn,english,aps,letterpaper]{revtex4}
\usepackage{amsmath, amssymb, amsthm, braket}
\usepackage{graphicx}
\usepackage{amsthm}
\theoremstyle{plain}
\newtheorem{theorem}{Theorem}[section]
\newtheorem{lemma}[theorem]{Lemma}

\usepackage[usenames,dvipsnames]{color}			
\definecolor{purple}{rgb}{0.5,0,0.5}

\begin{document}

\global\long\global\long\def\ket#1{|#1\rangle}
\global\long\global\long\def\bra#1{\langle#1|}
\global\long\global\long\def\proj#1{|#1\rangle\langle#1|}
\global\long\global\long\def\ketbra#1#2{|#1\rangle\langle#2|}
\global\long\global\long\def\braket#1#2{\langle#1|#2\rangle}

\title{Practical variational tomography for critical 1D systems}
\author{Jong Yeon Lee}
\email{jlee2@caltech.edu}
\affiliation{Institute for Quantum Information and Matter and Walter Burke Institute for Theoretical Physics, California Institute of Technology, Pasadena, California 91125, USA}
\author{Olivier Landon-Cardinal}
\email{olc@caltech.edu}
\affiliation{Institute for Quantum Information and Matter and Walter Burke Institute for Theoretical Physics, California Institute of Technology, Pasadena, California 91125, USA}

\begin{abstract}
We improve upon a recently introduced efficient quantum state 
reconstruction procedure targeted to states well-approximated by the 
multi-scale entanglement renormalization ansatz (MERA), e.g., ground states of critical models. 
We show how to numerically select a subset of experimentally accessible measurements which maximize information extraction about renormalized particles, thus dramatically reducing the required number of physical measurements. 
We numerically estimate the number of measurements required to characterize the ground state of the critical 1D Ising (resp. XX) model and find that MERA tomography on 16-qubit (resp. 24-qubit) systems requires the same experimental effort than brute-force tomography on 8 qubits. We derive a bound computable from experimental data which certifies the distance between the experimental and reconstructed states.
\end{abstract}

\maketitle


The understanding of many-body quantum systems has dramatically progressed recently, theoretically and experimentally. 
New efficient numerical methods use the properties of entanglement in many-body states, such as the area law of entanglement entropy~\cite{Hastings07}, 
to describe efficiently the many-body wave function of physical systems~\cite{Orus14}.
In parallel, experimentalists achieve a very high degree of control over larger and larger systems~\cite{HHR+05,MSB+11}.
However, efficient methods to quantitatively compare theoretical predictions to experimental realizations are few. 

Quantum state tomography~\cite{Hradil97} is a paradigm that aims to reconstruct the quantum state of a system
by performing multiple measurements on identically prepared copies of the system.
Since measurements perturb a quantum system, many copies of the system are needed  
to extract information about the many-body wave function.
Once the experimental data is extracted, a numerical procedure determines which density matrix fits best the measurements. This quantum state reconstruction can be performed using different approaches, the most used being maximum likelihood estimation~\cite{AJK04}.
   
Generally, both the number of measurements and the post-processing time of quantum state reconstruction
grows exponentially with the system size. This is not surprising since 
the dimension of the Hilbert space of $n$-particles grows exponentially in $n$. Note that we generically refer to the fundamental experimental objects of the physical system of interest as particles. For instance, for cold atoms in an optical potential, ``particles'' would correspond to cold atoms.
In an arbitrary many-body wavefunction, there is an exponential number of coefficients to estimate. 
Furthermore, for a (Haar)-random quantum state, most coefficients have exponentially small amplitudes in a local basis, so to distinguish any one of those amplitudes from zero, one must take an exponential number of samples. 
This simple reasoning hints towards an experimental and numerical efforts that scales exponentially with system size.

However, \emph{physical} quantum states, for instance ground states of local Hamiltonians, only constitute a very small subset of all states in the Hilbert space~\cite{PQS+11}. A general and very fruitful idea is to approximate those states of interest by a suitable variational family of states. An efficient family of states not only allows for a concise description of states --the number of parameters needed to represent them grows only polynomially with system size-- but also allows to efficiently compute physical quantities, such as expectation values of local observables in polynomial time. 

Tensor network (TN) states are variational families of states which are strong candidates to parametrize the physical part of the Hilbert space~\cite{Orus14}. 
TN states are built to accommodate the structure of entanglement for various physical states.
For instance, the matrix product states (MPS) representation is based on the property that the entropy of a block of particles grows with the boundary of the block. This property is called an area law, see~\cite{ECP10} for a review. Ground states of 1D gapped systems follow an area law~\cite{Hastings07} and are well-approximated~\cite{VC06} with matrix product states~\cite{AKL+87,FNW92,KSZ93,KSZ91} (MPS). Moreover, convenient numerical methods exist to find such a MPS approximation, such as the density matrix renormalization group (DMRG) method~\cite{White92,Schollwock11}. 
However, this \emph{area law} is violated by critical systems, i.e. ground states of quantum systems near a quantum phase transition~\cite{Sachdev01}. Indeed, in 1D critical systems, the entanglement of a block of $n$ particles diverges as $\log(n)$. To reproduce this entanglement scaling, Multi-Scale Entanglement Renormalization Ansatz (MERA) was introduced in~\cite{Vidal08}. A MERA state is the output of a specific type of quantum circuit whose gates arrangement generates an amount of entanglement. which grows logarithmically with block size. 

Recently, the use of variational states has been applied to tomography~\cite{BGC+13,CPF+10,LP12} and explicit state reconstruction methods have been given. The pioneering work on MPS tomography~\cite{CPF+10} provided the first demonstration that variational tomography could be performed efficiently. Subsequent work~\cite{LP12} demonstrated that variational tomography was also possible for 1D critical systems described by MERA states. MERA tomography offers the perspective to be an extremely valuable tool in the experimental characterization of quantum simulators, finely controlled systems which experimentalists can tune to reproduce the dynamics of a model Hamiltonian. Indeed, the MERA can be straightforwardly extended to study critical 2D models which are precisely the Hamiltonians that quantum simulators~\cite{BN09,CZ12} offer to probe experimentally.

In this paper, we revisit the idea of MERA tomography in 1D and explicitly investigate some of the challenges left open in the original proposal~\cite{LP12}. The original article gave a proof of principle that the tomography protocol only required a numerical and experimental effort scaling polynomially with the size of the system. Schematically, the idea is that the MERA transforms a highly-entangled state into a trivial product state by applying a logarithmic number of renormalization steps, each corresponding to a layer of gates in a quantum circuit. To identify each gate, it suffices to identify the density matrix on a block of renormalized particles of constant size. However, inferring information about renormalized particles is done through renormalized observables. These renormalized observables are accessed using measurements on the physical states and the knowledge of the previous renormalization steps. To maintain accuracy about the estimation value of a renormalized observable, the number of repeated measurements at the physical level is multiplied by a constant, the \emph{scaling factor}, for each renormalization step. Since there are only a logarithmic number of renormalization steps, the overhead in the number of measurements grows only polynomially with system size.

While this result is crucial theoretically, it does not guarantee that the number of measurement and the processing time is reasonable \emph{in practice} for moderately small systems within experimental reach. Indeed, the polynomial growth governs the asymptotic scaling in the limit of very large system size, but experimentalists are interested in the actual number of measurements required to characterize a system of interest. Thus, identifying the precise polynomial and in particular the power of the leading term, along with its constant multiplicative coefficient, is of paramount importance.  

The analysis of~\cite{LP12} focused only on the scaling factor $\lambda$ of single particle observables and found that the overhead in repeated physical measurements scaled as $\lambda \sim 6$ for the critical Ising model. This naively lead to an estimate of the total number of measurements needed that increases slowly with system size. However, this analysis failed to take into account that one needs to measure many-body observables for MERA tomography. For a ternary MERA on qubits, one needs to measure 5-body observables whose scaling factor is $\lambda^5$, resulting in an overhead on repeated measurements which is beyond experimental capacities, even for moderately large systems.

Here, we i) assess the reasons why the overhead in the number of measurements is much larger than naively anticipated, ii) suggest strategies to minimize it and iii) numerically demonstrate that those strategies lead to a reasonable total number of measurements for the critical Ising model on system size of experimental interest.  


The article is organized as follows. 
In section~\ref{sec:MERAtomo-recap}, we recall the idea of variational tomography focusing on MERA tomography. We discuss the concepts in the main body of the article. The technical discussion about our improved numerical algorithm for MERA tomography is available in the Appendix~\ref{sec:Numerical} for the interested reader.
In section~\ref{sec:ScalingMeasurements}, we investigate the scaling of the total number of experimental measurements needed to characterize an experimental state close to the groundtstate of a critical Ising 1D chain. We show that the naive approach of~\cite{LP12} requires an unreasonable (yet polynomial) amount of measurements. In section~\ref{sec:Improvement}, we suggest two possible solutions to resolve the issue. In section~\ref{numerical_obs}, we numerically show that the combinations of those two solutions significantly reduces the number of experimental measurements. In section~\ref{sec:propagation-errors}, we provide an analysis of the source of errors in our tomography scheme and infer the bound of the distance between experimental and reconstructed states, based on a more detailed analysis provided in the Appendix~\ref{sec:ErrorAnalysis}.
  

\section{MERA tomography}
\label{sec:MERAtomo-recap}

\subsection{Variational tomography}

The core idea of variational tomography is to take advantage of the succinct description of variational states in order to devise an efficient learning method. A learning method consists of three parts: i) the measurement prescription which identifies the measurements to perform, ii) the data acquisition when the measurements are performed and iii) the state reconstruction that infers the compatible quantum state via post-processing. Note that the measurement prescription can change adaptatively due to data acquisition as preliminary data can improve the choice of measurements. This is the case for MERA tomography. 

As mentioned in the introduction, the idea of variational tomography has been demonstrated on two variational class of states: MPS and MERA. In both cases, quantum state tomography is performed on small systems and numerical processing is used to stitch the density matrices of those small systems into a global state. While this stitching is efficient for both MPS and MERA, this procedure is expected to be very hard for arbitrary state in the Hilbert space. Recent progress has been made to understand the structure of quantum states for which local measurements are informationally-complete~\cite{Kim14}.

In MPS tomography~\cite{CPF+10}, reduced density matrices $\sigma_i$ on all blocks of a constant number (independent of system size) of particles are estimated. Then, a classical algorithm, inspired from ideas of compressed sensing, is used to reconstruct the global state. Alternatively, one could learn the quantum circuit preparing the state. Indeed, any MPS can be prepared using a staircase circuit with linear depth (see Fig.~\ref{fig:MPS}). 
\begin{figure}
 \includegraphics[width=0.9\columnwidth]{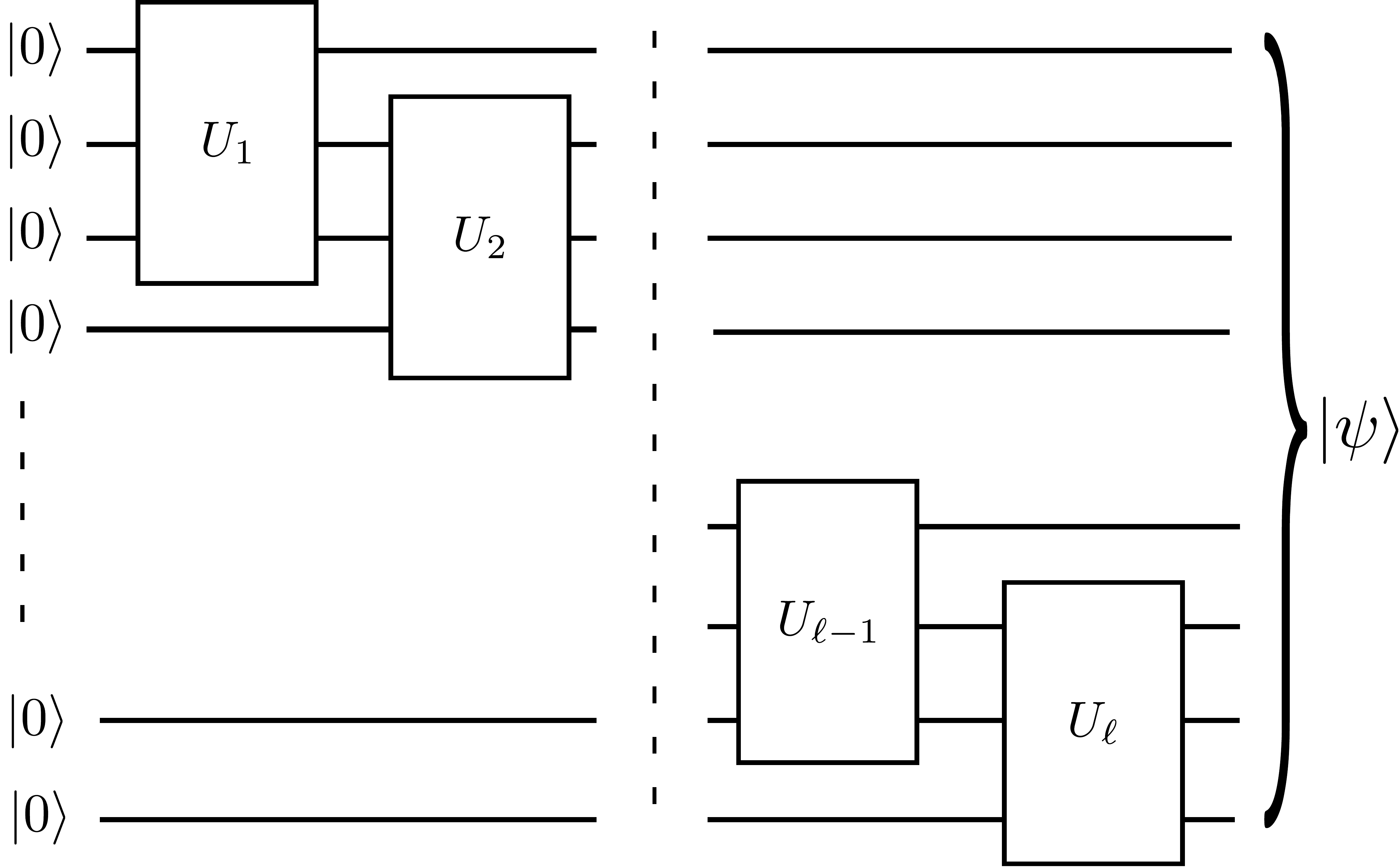}
	\caption{\label{fig:MPS} A matrix product state is obtained from a product state by applying a staircase circuit. This structure allows to sequentially infer the quantum gates.}
\end{figure}

We now describe MPS tomography in more details as it bears many similarities with MERA tomography, which will be discussed in the next section. 
One can learn the rightmost quantum gate $U_\ell$ of Fig.~\ref{fig:MPS} by performing tomography on a small number of particles and then identifying a unitary gate which disentangles the bottom particle and puts it in the state $\ket{0}$. 
We then repeat this procedure on state $U_\ell^\dagger \ket{\psi}$.
To learn $U_{\ell-1}$, the original proposal of~\cite{CPF+10} was to experimentally apply the gate $U_\ell^\dagger$. However, one can use the knowledge of $U_\ell$ to see how it modifies the physical observables on the physical state $\ket{\psi}$. In other words, the knowledge of $U_\ell$ allows us to translate measurements on the physical state $\ket{\psi}$ into what would be obtained by performing the measurement of some \emph{renormalized} observables on the renormalized state $U_\ell^\dagger \ket{\psi}$. Each physical observable will be associated to a renormalized observable. As long as renormalized observables span the support of the density matrix, they are informationally complete. The power of renormalized observables was not immediately realized in~\cite{CPF+10}, but became apparent when MERA tomography was devised~\cite{LP12}. We now describe MERA tomography in great details as our work builds upon it.

\subsection{Learning MERA states}

\subsubsection{Quantum circuit for MERA states}
The MERA is a variational family of states~\cite{Vidal08} arising from a real-space renormalization group approach called entanglement renormalization~\cite{Vidal07}.
Entanglement renormalization creates a sequence of quantum states $\{\rho_\tau\}_{\tau=0\dots T}$ where $\rho_0$ is the physical state (which we will also refer to as the experimental state in the context of tomography) and $\rho_{\tau >0}$ are coarse-grained version of the physical state which encode entanglement on a larger scale. Intuitively, one can think of each renormalized state $\rho_\tau$ as a state of a 1D chain of $n/(k^\tau)$ spins where $k=2$ for binary MERA and $k=3$ for ternary MERA. The crucial insight of MERA is that for critical states, it is important to get rid of short scale entanglement before each renormalization step. Otherwise, the short scale entanglement accumulates and the renormalization cannot be carried anymore. This renormalization approach translates into a quantum circuit, depicted on Fig.~\ref{fig:MERA} that turns the physical state $\rho_0$ on $n$ particles into the all zero state $\ket{0}^{\otimes{n}}$ (in the case of a pure state).   

\begin{figure}
	\includegraphics[width=0.9\columnwidth]{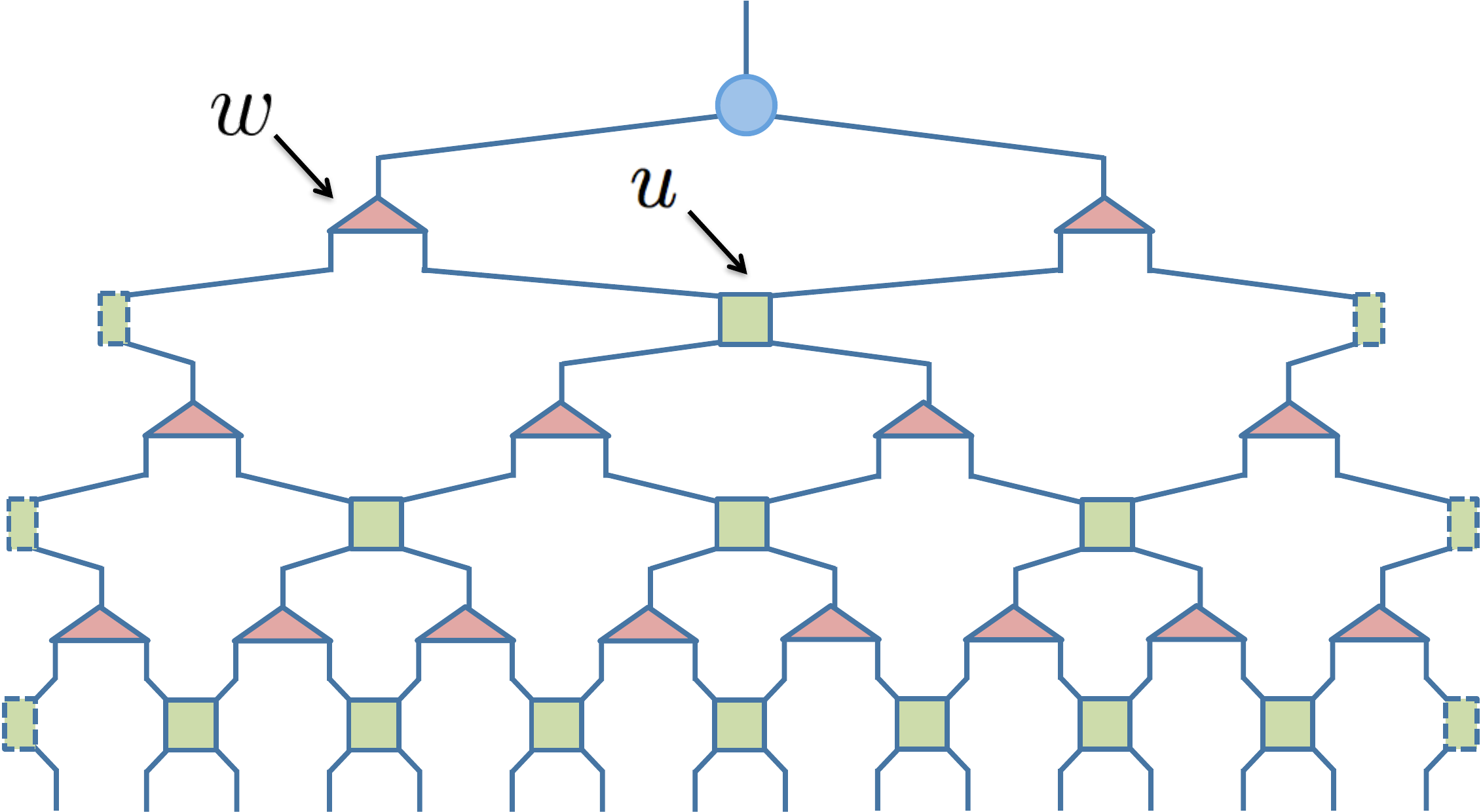}\\
	\vspace{0.1in}
	\includegraphics[width=0.9\columnwidth]{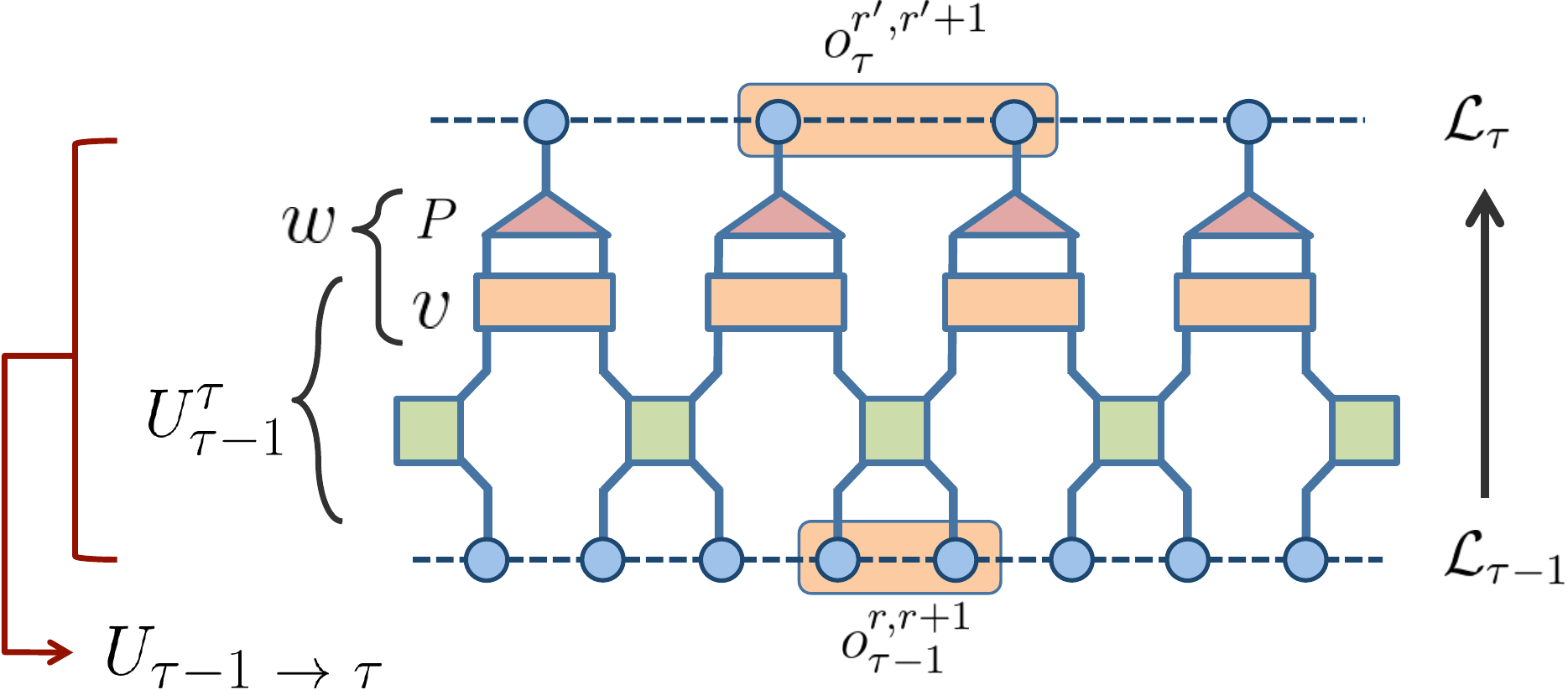}
	\caption{\label{fig:MERA} (Color online) {\bf (up)} Example of binary 1D MERA for a lattice of 16 particles. (periodic boundary) \textit{u} is a disentangler and \textit{w} is an isometry. {\bf (down)} An isometry \textit{w} can be decomposed into unitary \textit{v} followed by a projector \textit{P}.
	$U_{\tau-1 \to \tau}$ is the isometry of layer $\tau$ which coarse-grains a state at level $\tau-1$ to a state at level $\tau$ whereas $U_{\tau-1}^\tau$ is a the unitary part of $U_{\tau-1 \to \tau}$. $o_{\tau-1}$ is an operator at level $\tau-1$ and it is mapped into $o_{\tau}$ at level $\tau$. }
\end{figure}

This MERA circuit consists of two sets of quantum gates. 
The \emph{disentanglers} are unitary transformations, depicted by squares on Fig.~\ref{fig:MERA} and denoted $u$, whose goal is to remove short scale entanglement. 
The \emph{isometries}, depicted by triangles on Fig.~\ref{fig:MERA} and denoted $w$, map several particles into a single renormalized particle by applying a unitary transformation $v$ followed by projection operator $P$, see bottom of Fig.~\ref{fig:MERA}. For instance, in binary MERA, two particles whose individual quantum dimension is $\chi$, i.e., whose total quantum dimension is $\chi^2$, are mapped into a single particle of quantum dimension $\chi$. 
Note that this transformation is only possible if the density matrix before the isometry is (approximately) supported on a space of dimension $\chi$ rather than having full rank $\chi^2$. In other words, the purpose of the disentanglers is precisely to locally rotate the Hilbert space to concentrate the support of the density matrices. This remark is at the heart of the numerical method to identify disentanglers.

Another important notion of a MERA circuit is the past causal cone of a quantum gate and the future causal cone of particles. Imagine that time flows from the bottom of Fig.~\ref{fig:MERA} to the top. In other words, the \emph{level} index $\tau$ plays the role of time. Level-0 corresponds to the physical state while $\tau>0$ indices the states and lattices obtained after $\tau$ step of renormalization. The transformation $U_\tau$ from $\rho_{\tau-1}$ to $\rho_{\tau}$ corresponds to a \emph{layer} of quantum gates, see Fig.~\ref{fig:MERA}. For any given quantum gate of the circuit (disentangler or isometry), its past causal cone is the set of physical particles whose change would induce a change of the quantum gate. For any set of physical particles, its future causal cone is the set of quantum gates such that a particle belongs to the past causal cone of at least one of the gates in the set.  



\subsubsection{MERA tomography procedure}

Let us briefly describe the MERA tomography procedure, taking the binary MERA geometry (see Fig.~\ref{fig:MERA}) as example. The goal is to find a MERA circuit representing a given experimental state. To do this, MERA tomography repeatedly measures local observables to obtain the reduced density matrices of 4 renormalized particles in lattice $\mathcal{L}_{\tau}$ which are the past causal cone of each isometry mapping $\mathcal{L}_\tau$ to $\mathcal{L}_{\tau+1}$ (see Fig.~\ref{fig:optimization_geometry}). We will often refer to renormalized particles as sites on the renormalized lattices. Hence, a density matrix on 4 renormalized particles will be referred to as a 4-site density matrix. 

For the physical level, i.e., $\tau=0$, those reduced density matrices are obtained by brute-force quantum state tomography. This is efficient since brute-force tomography is only performed on a block of constant size. For higher layers, the density matrices can be inferred from physical measurements and the knowledge of the quantum gates in the previous layers. We will describe the procedure in more details in Sec.~\ref{sec:ScalingMeasurements}. For the moment, let us assume that we know every 4-site density matrix $\rho_i$ corresponding to the past causal cones of every isometry $i$ in layer $\tau+1$ .
\begin{figure}
 \includegraphics[width=0.8\columnwidth]{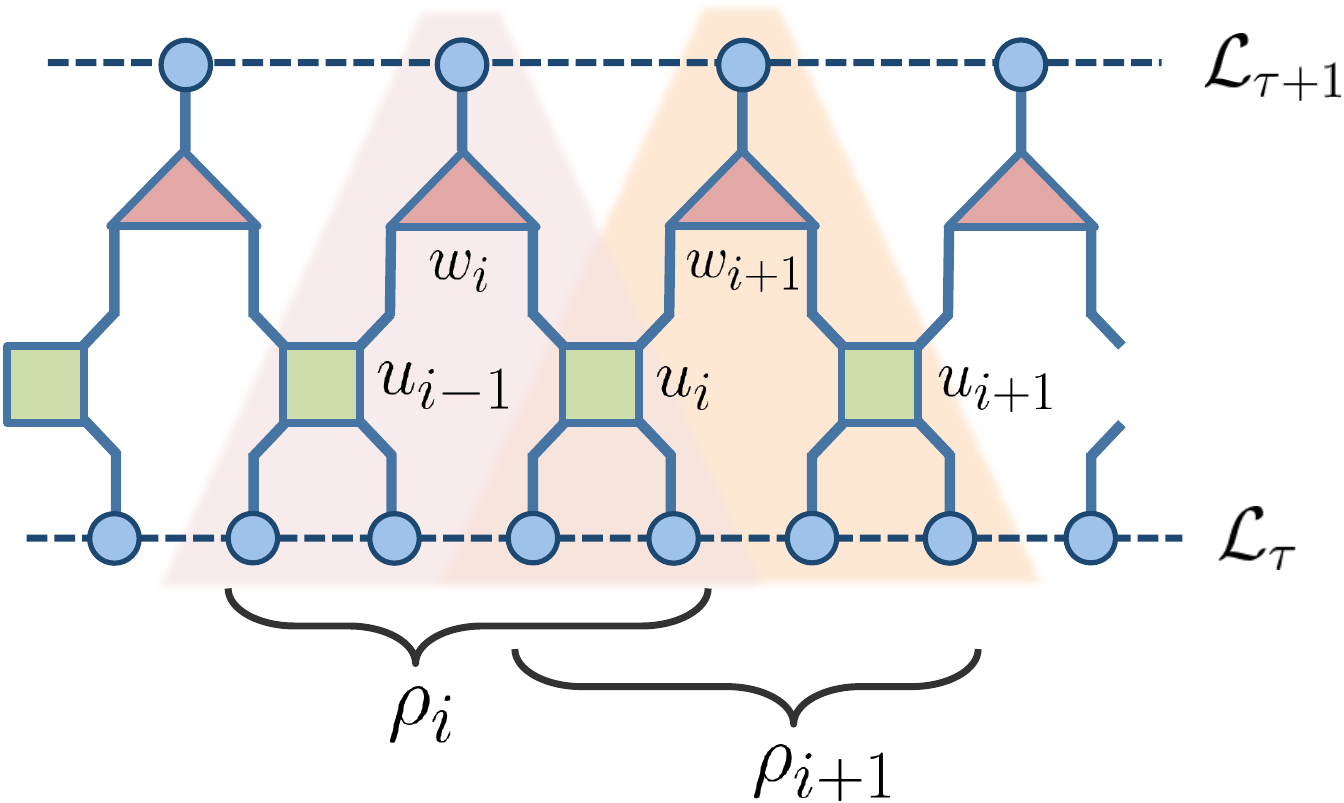}
 \caption{\label{fig:optimization_geometry} (Color online) Each past causal cone of the isometry $w_{i}$ is the 4 sites in state $\rho_i$. The choice of the disentangler $u_i$ affects both isometries $w_{i}$ to its left and $w_{i+1}$ to its right. }
\end{figure}

Given $\left\lbrace\rho_i^\tau\right\rbrace_i$, the goal is to find the disentanglers in layer $\tau+1$. Let us focus on a single disentangler $u=u_i$, supposing that all other disentanglers in the layer are fixed. The choice of $u$ will affect the isometries, $w_{i}$ and $w_{i+1}$, respectively to the left and the right of $u$, see Fig.~\ref{fig:optimization_geometry}.

Thus, the objective function $g$ splits into two parts 
\begin{equation}
 g(u,\rho_i^\tau,\rho_{i+1}^\tau)=f_L(u,\rho_{i}^\tau)+f_R(u,\rho_{i+1}^\tau)
\end{equation}
where $f_L(u,\rho_{i})$ corresponds to minimizing the rank of the 2-site reduced density matrix that is the input of $w_{i}$ and similarly for $f_R$ with respect to $w_{i+1}$. 

After applying the optimal disentanglers, the two-site reduced density matrix at the input of the isometry $w_i$ should have a rank at most $\chi$ so that the isometry keeps the $\chi$ eigenvectors with largest eigenvalues. In other words, we want the probability weight to be supported on the $\chi$ largest eigenvalues. Thus, we maximize the objective function,
\begin{equation} 
f_{L,R}(u,\rho_{i}^\tau) = \sum_{k \leq \chi} \lambda_{k} \label{object1}
\end{equation}
where $\rho_{i}^\tau$ is the reduced density matrix for the $i$-th block at level $\tau$ and $\lambda_{k}$ is $k$-th eigenvalue of the reduced density matrix after the disentangler $u$ has been applied. 

Once all disentanglers have been obtained, the isometries $w_i$ are obtained by diagonalizing the reduced density matrices $\sigma_i=\text{tr}_{14}\left[\rho_i\right]$ at the input of the isometries where $\text{tr}_{14}$ implies tracing over site $i_1$ and $i_4$ after disentanglers (see Fig.~\ref{fig:TNC}). Indeed, one can decompose the isometry $w_i$ as an unitary transformation $v_i$ followed by a projector $P$ of rank $\chi$. Given the diagonalization
\begin{equation}
 \sigma_i=\sum_{k\leq \chi}\lambda_k \proj{\phi_k} + \sum_{k> \chi}\lambda_k \proj{\phi_k}
\end{equation}
the unitary $v_i$ maps the first $k$ eigenvectors $\ket{\phi_k}$ to $\ket{k}\otimes\ket{0}$. The way it acts on the other eigenvectors is arbitrary, as long as $v_i$ is unitary. Afterwards, the projector $P=\mathbb{I}_\chi\otimes\proj{0}$ throws away the irrelevant eigenvectors.
This procedure is repeated over each layer of the MERA circuit. 

In the original MERA tomography procedure described in~\cite{LP12}, a conjugate gradient method was used to maximize the objective function given by Eq.~\eqref{object1}. In this paper, an alternative approach inspired by~\cite{EV09} was used for this maximization. This numerical procedure is discussed in details in Appendix~\ref{sec:Numerical}


\section{Scaling of the number of experimental measurements}
\label{sec:ScalingMeasurements}

\subsection{Ascending superoperator}

As explained in Sec.~\ref{sec:MERAtomo-recap}, MERA tomography infers the quantum circuit preparing the experimental state from a product state. 
To identify each gate, the numerical procedure takes as input the reduced density matrix on a small block of particles. For the physical layer, denoted $\mathcal{L}_0$ on Fig.~\ref{fig:MERA}, those particles correspond to experimentally measurable particles. However, this is not the case, for higher renormalized levels, $\mathcal{L}_\tau$ for $\tau>0$. 
To get access to the density matrices on block of \emph{renormalized} particles, we will assess how physical measurements will be mapped into effective measurements at higher levels. 
This mapping depends on the disentanglers and isometries between the physical level $\mathcal{L}_0$ and the current level $\mathcal{\tau}$.  Thus, it depends on the information acquired by tomography on the previous layers.

\subsubsection{First layer of renormalization}

Let us consider the first layer of renormalization. Let's define $U_0^1$ as the product of all the disentanglers $u$ and all unitary transformations $v$ (see Fig.~\ref{fig:MERA}). Note that $U_0^1$ is a unitary transformation since it does not contain $P$, the projection part of isometries, which reduces the dimension of the Hilbert space. Thus, before truncation, the observable $O_0$ at the physical level is mapped to the semi-renormalized observables $U_0^1 O_0 (U_0^1)^\dagger $ since 
\begin{equation}
\text{tr}(\rho_0 O_0) = \text{tr}( U_{0}^{1} \rho_0  ( U_{0}^{1})^\dagger  U_{0}^{1} O_0 ( U_{0}^{1})^\dagger) \\
\end{equation}
where $\rho_0$ is a density matrix at the physical level. 

However, the crucial step of the renormalization scheme is to reduce the dimension of the Hilbert space. Formally, the idea is that $\tilde{\rho}_1=U_{0}^{1} \rho_0  (U_{0}^{1})^\dagger$ is not full rank but has the form $\tilde{\rho}_1=\rho_1 \otimes |00\dots 0\rangle \langle00\dots 0|$. Thus, one can keep only the relevant degrees of freedom by applying a projector $P$ which removes the superfluous degrees of freedom (see Fig.~\ref{fig:MERA}), i.e., 
\begin{equation}
 P \tilde{\rho}_1 P^\dagger = \rho_1 
\end{equation}

Hence, the expectation value of the physical operator $O_0$ can be written as
\begin{eqnarray}
\text{tr}(\rho_0 O_0) & = & \text{tr}(P \tilde{\rho}_1 P^\dagger P U_{0}^{1} O_0 (U_{0}^{1})^\dagger P^\dagger) \\
& = & \text{tr}(\rho_1 \mathcal{A}_{0}^{1}\left[O_0\right]) 
\end{eqnarray}
where $\mathcal{A}_{0}^{1}\left[O_0\right]$ is a renormalized observable. The action of an ascending superoperator $\mathcal{A}_{0}^1$ is defined by
\begin{equation}
 \mathcal{A}_{0}^1[...]=P U_{0}^{1} [...] (U_{0}^{1})^\dagger P^\dagger.
\end{equation}

\subsubsection{Multiple layers of renormalization}

The reasoning to go from the physical level to the first renormalized level can be iterated. In that way, one defines an ascending superoperator from level $0$ to level $m$ $\mathcal{A}_{0}^m=\prod_{k=1}^m \mathcal{A}_{k-1}^k$, which maps operators at the physical level $O_0$ to operators acting at level $m$ obeying the equation
\begin{equation} \label{eq:scaleO}
\text{tr}(\rho_0 O_0)= \text{tr}(\rho_m \mathcal{A}_0^{m}\left[O_0\right])  
\end{equation}

Eq.~\eqref{eq:scaleO} allows us to relate the renormalized state $\rho_m$ to the measurements $\text{tr}(\rho_0 O_0)$ once we know the ascended observable $\mathcal{A}^m_0\left[O_0\right]$.

We can express the superoperator $\mathcal{A}^m_0$ as a matrix $M_{ij}$ by choosing bases of observables $\left\{O_0^i\right\}$ at the physical level and $\left\{O_m^j\right\}$ at level $m$. Inferring the physical measurement corresponding to an effective measurement on renormalized particles then reduces to inverting this matrix to get $M^{-1}$, 

\begin{eqnarray}
{\cal A}^m_0(O^j_0) =& \sum_i M_{ij} O^i_m   \\ 
\text{tr}(\rho_m O^i_m) =&\sum_i M^{-1}_{ij} \text{tr}(\rho_0 O^j_0)
\end{eqnarray}

\subsection{Overhead on the number of physical measurements}

The strategy to infer information about the renormalized state at level $m$ is now clear: one performs measurements $O_0$ at the physical level and then use the knowledge of the gates in the circuit to compute the ascending superoperator $\mathcal{A}_0^m$ and thus the renormalized observables $\mathcal{A}_0^m\left[O_0\right]$. 

Let us illustrate this approach for scale-invariant MERA in critical systems. In that case, translation-invariance and scale-invariance guarantee that isometries and disentanglers at all levels and sites are the same, which means that the scaling behavior of operators does not depend on the level considered. Moreover, in ternary 1D MERA, we can use one-site physical operators which are mapped into one-site renormalized operators (see Fig.~\ref{fig:scale}). If $O_0$ has a support in that site, then the tensor network contraction for ${\cal A}(O_0)$ can be simplified as in Fig.~\ref{fig:scale}. This simplifies the tomography procedure. Indeed, to calculate the scaling, we only need information about the isometry $w$. 

\begin{figure}
	\includegraphics[width=0.75\columnwidth]{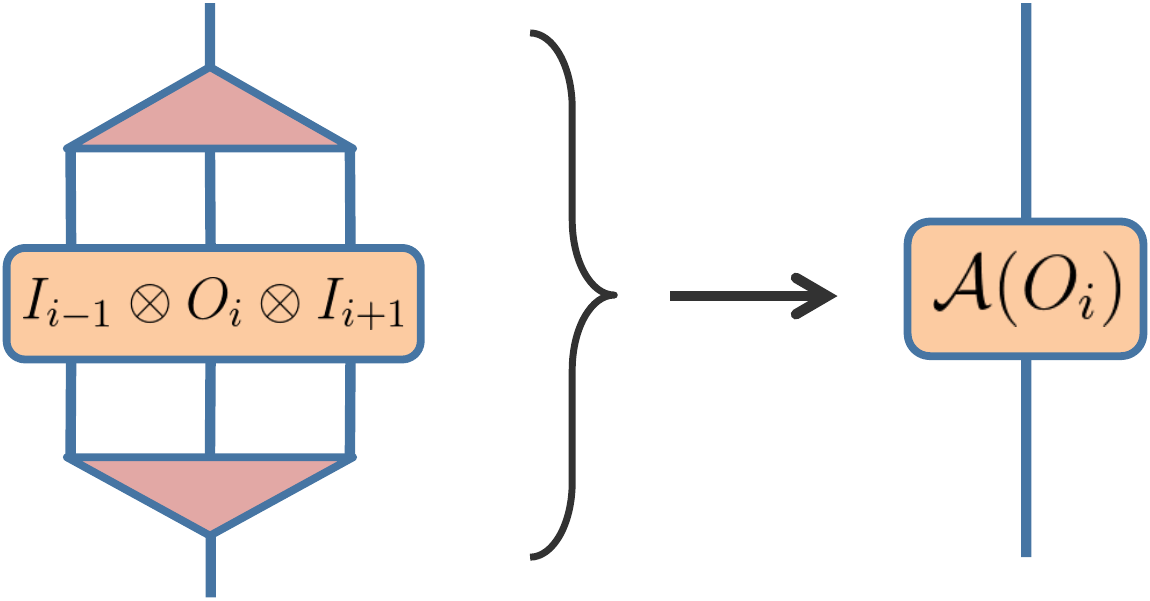}
	\caption{\label{fig:scale} (Color online) Ascending super-operator and renormalized observable for a ternary MERA. The tensor network contraction turns a single-site operator $O_i$ at level $\tau$ into a single-site operator $\mathcal{A}(O_i)$ at level $\tau+1$. }
\end{figure}

We studied a few 1D critical models including Ising, XX, and Potts using a ternary MERA code to study the scaling behaviors of observables. 
Let us focus on the case of the critical Ising model. Choosing the Pauli basis, $\{ O^i \} = \{I,\sigma_x, \sigma_y, \sigma_z\}$ for observables, the matrix representation of the descending superoperator  $ M^{-1}$ reads


\begin{equation} \label{Gmatrix}
\left( M^{-1}_{ij} \right) = \begin{pmatrix}
 1& 1.1 & 0 & 1.7  \\ 
 0&  2.01 & 0& 1.55\\ 
 0&  0& 2.41 &0\\ 
 0&  0.3 & 0&	2.41\\ 
\end{pmatrix} 
\end{equation}

Let's focus on the observable $\sigma^y$ which is an eigenvector of the ascending superoperator since
\begin{equation}
 \mathcal{A}_0^1\left[\sigma_0^y\right]=\frac{1}{\sqrt{\lambda_y}} \sigma_1^y
\end{equation}
where $\sqrt{\lambda_y}=2.41$.

Using Eq.~\ref{eq:scaleO}, one gets that
\begin{equation}
 \text{tr}(\rho_1 \sigma_1^y) = \sqrt{\lambda_y} \text{tr}(\rho_0 \sigma_0^y)
\end{equation}

One crucial point to worry about is the statistical error on the expectation values due to the finite number of measurements. 
Due to statistical fluctuations, the measured expectation value $\langle \sigma_0^y \rangle_{\rho_0}$ will be equal to the proper expectation value $\text{tr}(\rho_0 \sigma_0^y)$ up to some error $\epsilon_0$ which scales like $N_0^{-1/2}$ where $N_0$ is the number of repeated measurements, i.e., 
\begin{equation}
 \langle \sigma_0^y \rangle_{\rho_0} = \text{tr}(\rho_0 \sigma_0^y) \pm \epsilon_0
\end{equation}
When inferring the expectation value $\text{tr}(\rho_1 \sigma_1^y)$, the uncertainty will also be multiplied

\begin{eqnarray}
 \langle  \sigma_1^y \rangle_{\rho_1} & = & \sqrt{\lambda_y} \langle  \sigma_0^y \rangle_{\rho_0} \\
 & = & \sqrt{\lambda_y} \text{tr}(\rho_0 \sigma_0^y) \pm \sqrt{\lambda_y} \epsilon \\
 & = & \text{tr}(\rho_1 \sigma_1^y) \pm \sqrt{\lambda_y} \epsilon
\end{eqnarray}
Thus, to maintain the accuracy $\epsilon_0$ at the renormalized level, one needs to perform a number of measurements

\begin{equation}
 N=\lambda_y N_0
\end{equation}
More generally, if $O^j_1$ is not an eigenvector of the ascending super-operator, the total number of measurements need to be multiplied by $\sum_i | M^{-1}_{ij} |^2 $.

This overhead in the number of measurements (i) will multiply with the number of particles unto which the observables act non-trivially and (ii) will multiply between each layers. From a theoretical point of view, points (i) and (ii) are not catastrophic since they only correspond to a polynomial overhead. Indeed, for point (i), the number of particles in a tomography block is a constant, independent of system size. For point (ii), the overhead depends on system size but is only polynomial. To go from level $0$ to level $m$, the multiplicative factor will be $\lambda_0^m=\prod_{k=1}^m \lambda_{k-1}^k$ but there is only a logarithmic number of layers in the MERA circuit. Thus, from an asymptotic scaling point of view, the method induces only a polynomial overhead. However, for finite size system of interest, this overhead on the number of physical measurements can be dramatic. We will now see that the naive approach outlined here leads to overhead which is unreasonable for experimentalists, before suggesting two improvements that will keep the total number of measurements reasonable.

 \subsection{Prohibitive experimental cost for 1-site observables}
 
 Returning to the example of the Ising model at criticality, we see from Eq.~\eqref{Gmatrix} that maintaining the accuracy at the renormalized level requires $\lambda\approx6$ times the number of measurements than the one at the  physical level. However, this analysis is appropriate only for one-site observable. This fact, which had not been appreciated in~\cite{LP12}, has dramatic consequences. 
 
 Let us now briefly describe a way to perform brute-forve tomography, before returning to MERA tomography.
 In order to estimate the expectation value of observables, a practical method is the so-called $3^n$ method~\cite{AJK04} where one measures observables which are tensor product operators by measuring each individual operators on the same copy of the system and then post-processing classically the information. For instance, suppose we are interested in a chain of qubits and want to estimate the expectation value of $\sigma^z_1 \otimes \sigma^z_2$, an operator which acts non-trivially, but as a tensor product, on qubits 1 and 2. Rather than measuring $\sigma^z_1 \otimes \sigma^z_2$ at once, we can measure $\sigma^z_1$, record the eigenvalue $s_1$ we measured and then measure $\sigma^z_2$ on the same copy and record the eigenvalue $s_2$. That way, we get a sample not only of $\sigma^z_1 \sigma^z_2$, but also some information on $\sigma^z_1 \otimes \mathbb{I}$. This method only requires to perform measurements of $3^n$ operators which are non-trivial on all $n$ particles of the chain, rather than $4^n$. Of course, there is additional information in the partial measurements. The key property here is that observables are tensor product of single-body observables.
 
 In MERA tomography using the ternary geometry, we can use this procedure since renormalized observables can be chosen to be tensor product of renormalized single-body observables. To measure renormalized operators, one only needs to measure the corresponding physical observables on the physical state. However, the number of repeated measurements will be multiplied by $\lambda$, for each single-body observable. Thus, for an observable which is the tensor product of 5 single-body observables, the overhead $S_{block}$ in the number of repeated measurements for block $S_{block}$ will be 
 \begin{equation}
  S_{block}=\lambda^5
 \end{equation}
Thus, to obtain the five-site reduced density matrix on renormalized particles while keeping the same precision as brute-force tomography on physical particles, the number of repeated measurements is multiplied by $\lambda^5$.
 
Furthermore, for a block of renormalized particles at the third layer, the multiplicative overhead is $(\lambda^2)^5$. For $\lambda = 6$ as in our case, this amounts to $6^{10}\simeq 6\times 10^7$. This lead to an unreasonable overhead on the number of measurements for experimentalists. Thus, the approach for MERA tomography needs to be improved in order to be of practical interest for experimentalists.

In the next section, we will suggest two improvements to limit the overhead on the number of measurements. We will then see in Sec.~\ref{sec:ScalingMeasurements} that those improvements dramatically reduce the overhead for the critical Ising model.

\section{Improved approach for MERA tomography}
\label{sec:Improvement}

\subsection{Optimizing the choice of physical observables}

While using tensor product of physical observables which are eigenvectors of the ascending superoperator was appealing from a theoretical point of view, this choice leads to an unreasonable number of measurements for accuracy. Instead, one can vary over the physical observables and select a subset of them which maximizes information extraction. From a tomography point of view, any set of physical observables whose renormalized versions span the space of density matrices on the renormalized block is admissible. We expect that many sets of admissible physical observables exist since the number of physical particles in the past light cone of a renormalized block is much larger than the number of renormalized particles in the block. The problem thus becomes to pick the optimal admissible set.


Of course, one needs to vary over physical observables which are experimentally accessible. In the case of qubits, we restrict ourselves to Pauli observables, i.e., tensor product of Pauli operators. When we map those Pauli observables to renormalized operators, they will become a set of non-orthogonal operators, each of which has different length and direction in the operator space. Among those renormalized operators, we can find a set of operators that give maximum information about the renormalized layer. 

\begin{figure}
	\includegraphics[width=0.95 \columnwidth]{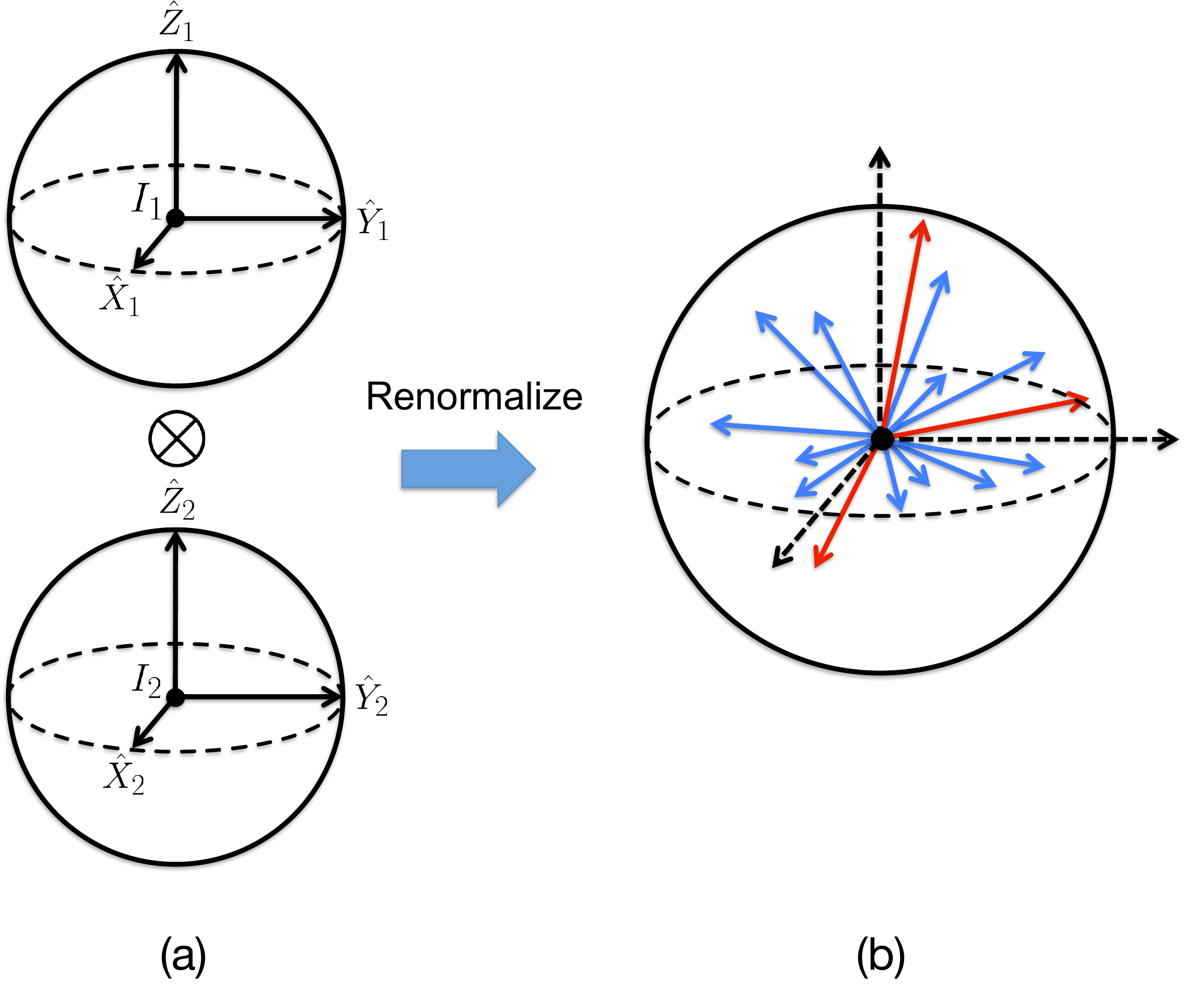}
	
	\caption{\label{fig:example} (Color online) Schematic diagram for the selection of optimal renormalized operators. In this example, two qubits are mapped unto a single renormalized qubit. 
	Without loss of generality, we choose the Pauli observables as a basis of observables for each qubit. The Pauli observables are represented as arrows on the Bloch spheres in a).
	By taking tensor product, we obtain 16 orthogonal operators. Each of those 16 operators will be mapped to a renormalized operator. 
	The identity operator (represented by big dot at origin) is mapped to the identity operator, but all the other 15 operators are mapped into some renormalized operators with different directions and magnitude, which are represented on the Bloch sphere of the renormalized qubit in b). Our task is to find 4 renormalized operators which span the renormalized Hilbert space and which are the most efficient in tomographic procedures. Intuitively, this optimal choice correspond to the 4 renormalized operators whose determinant is the largest, which are represented in red in b).	}
\end{figure}

This procedure is schematically explained in Fig.~\ref{fig:example}. 
In this example, we assume the physical Hilbert space with two qubits is renormalized into a Hilbert space with one qubit. Without loss of generality, we choose the Pauli observables $\mathcal{P}=\{I_k,X_k,Y_k,Z_k\}$ as a basis of observables for each qubit $k\in\{1,2\}$. By taking tensor product, we obtain 16 orthogonal operators of the form $\{ O^i_1 \otimes O^j_2\}$ where $O^i_1,O^j_2\in \mathcal{P}$. Each of those 16 operators will be mapped to a renormalized operator  $\{ {\cal A}(O^i_1 \otimes O^j_2) \}$. Each of those 16 operators will be mapped to a renormalized operator  $\{ {\cal A}(O^i_1 \otimes O^j_2) \}$. 

In order to span the renormalized Hilbert space, we only need four renormalized operators out of the sixteen available renormalized operators. 
Since ${\cal A}(I_1 \otimes I_1) = I$, we already have the renormalized identity operator so we need three more. 
Along with the identity operator, the three additional renormalized operators need to span the renormalized Hilbert space. 
Furthermore, we would like them to have a large determinant so that they cover the renormalized Hilbert space ``well'', 
in the sense that an arbitrary state in the renormalized Hilbert space can be reconstructed tomographically by a small number of repeated measurements. 
Thus, we choose the most informationally efficient set of operators to be the one with maximal determinant. 
In the example of Fig.~\ref{fig:example}, the set of operators with red-colored arrows maximize the determinant.
As we will see in Sec.~\ref{sec:ScalingMeasurements}, we will face the problem of renormalizing operators on 8 qubits into operators on 4 qubits, i.e., we will have to choose $4^4$ observables out of $4^8$.
The task of choosing the set of operators with maximal determinant turns out to be numerically intensive. 
In Sec.~\ref{numerical_obs}, we will introduce a heuristic to perform this task and show that this approach significantly reduces the overhead on the number of physical measurements.


\subsection{Changing the MERA geometry}

Another possible improvement to MERA tomography is to use the binary MERA geometry rather than the ternary MERA geometry. The ternary MERA geometry is unfavorable since it requires to identify the 5-site reduced density matrix in the past light cone of each isometry, while for the binary MERA, one needs to identify only 4-site reduced density matrix (see Fig.~\ref{fig:optimization_geometry}). This can make a significant difference on the number of measurements needed. 

Moreover, the binary MERA geometry has a structure which well-suited to apply the algorithm to select the optimal set of renormalized observables. 
Indeed, the past light cone of the 4-site reduced density matrix at level $\tau$ is 10 sites at level $\tau-1$ for the binary MERA geometry, much fewer than the 17 sites required in the ternary MERA geometry. Thus, we chose to select among the $4^{10}$ Pauli observables on 10 qubit a subset of $4^4$ which give maximal information about the 4-qubit density matrix at the next level using a heuristic which will be presented in details in Sec.~\ref{Greedy algorithm}. Numerically, we found that restricting the Pauli observables to act only on the 8 qubits indicated on Fig.~\ref{fig:binary} gave satisfactory results and made the running time and memory requirements of the heuristic more reasonable. In the next section, we describe the numerical results obtained by optimizing the choice of physical observables on a binary MERA geometry.

\begin{figure}

	\includegraphics[width=0.95 \columnwidth]{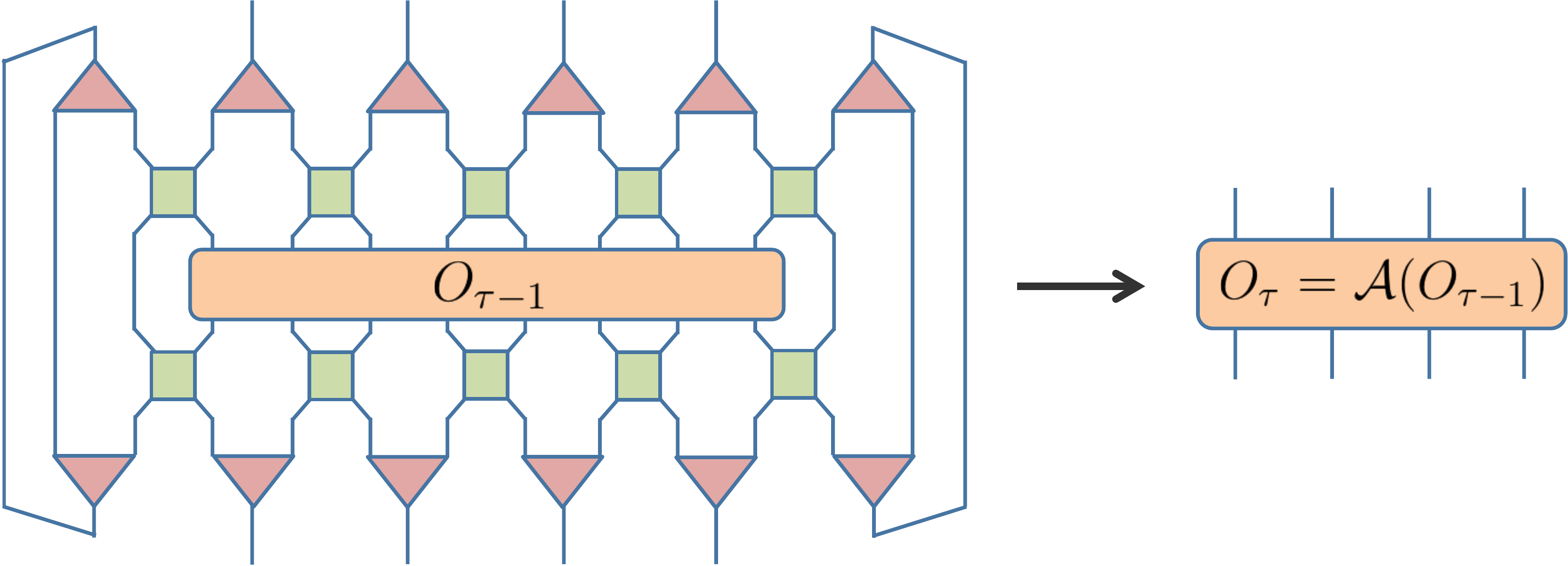}
	
	\caption{\label{fig:binary} (Color online) Tensor contraction in a binary MERA geometry. The observable $O_{\tau-1}$ acting on 8 sites at level $\tau-1$ is renormalized into a four-sites operator $O_{\tau}=\mathcal{A}(O_{\tau-1})$.}
\end{figure}


\section{Numerical results}\label{numerical_obs}

\subsection{Optimizing the choice of physical observables}

To optimize the choice of physical observables, we used a heuristic approach. We tested our approach on a 24-qubit ground state of the critical Ising and XX models. The state used to represent the experimental state is a $\chi=2$ binary MERA approximation to the ground state, which is obtained by a MERA energy minimization program.

\subsubsection{Greedy algorithm to maximize the determinant} \label{Greedy algorithm}

Given the disentanglers and isometries between levels $\tau-1$ and $\tau$ (which would have been identified thanks to tomography procedures), we calculated the $4^8$ renormalized operators corresponding to Fig.~\ref{fig:binary}. The task is now to choose a subset of $4^4$ renormalized operators that (i) span the space of the 4-qubit density matrix and (ii) span it in a way that maximizes the information acquisition (and thus minimize the number of repeated measurements). Criterion (ii) would be interesting to investigate from a theoretical point of view. In our work, we chose to \emph{maximize the absolute value of the determinant} of the set of renormalized operators as a proxy to maximizing the information acquisition. The intuition is that a large determinant will correspond to a set of renormalized operators which spans well the space of the 4-qubit density matrix.

To maximize the determinant, we used the following heuristic. We first chose the renormalized observable with the largest norm (choosing the norm induced by the Hilbert Schmidt inner product). Then, we vary over the remaining renormalized observables to find one that maximizes the determinant with the first one. We repeat this procedure over and over, obtaining a greedy algorithm to select the $4^4$ renormalized operators. This algorithm, named `Longest residual vector selection (LRV)' in~\cite{IZ04}, is one approach for the classic signal processing problem called \emph{matching pursuit}. 

The LRV algorithm is a heuristic which can be suboptimal. Let's illustrate such a situation by considering a simple two-dimensional space spanned by the orthonormal vectors $\hat{e}_1$ and $\hat{e}_2$. Consider the candidate set $\{ \hat{e}_1,\frac{(1-\epsilon)}{\sqrt{2}}(\hat{e}_1 + \hat{e}_2), \frac{(1-\epsilon)}{\sqrt{2}}(\hat{e}_1 - \hat{e}_2)\}$ where $\epsilon>0$ is small. We want to choose 2 vectors which maximize the absolute value of the determinant. By inspection, the best choice is $\{\frac{(1-\epsilon)}{\sqrt{2}}(\hat{e}_1 + \hat{e}_2), \frac{(1-\epsilon)}{\sqrt{2}}(\hat{e}_1 - \hat{e}_2)\}$ which has determinant $(1-\epsilon^2)$. However, the LRV algorithm will first select $ \hat{e}_1$ which has maximal norm and then select either one of the two remaining vector resulting in he choice $\{ \hat{e}_1,\frac{(1-\epsilon)}{\sqrt{2}}(\hat{e}_1 + \hat{e}_2)\}$ which has determinant $\frac{1-\epsilon}{\sqrt{2}}$. For a non-zero small $\epsilon$, the choice made by the LRV algorithm is dramatically worse than the optimal choice. 
In~\cite{IZ04}, an algorithm called `one by one replacement' is introduced to improve a (suboptimal) set of vectors by iteratively identifying bad choices in the current set and replacing it by a better vector from the candidate set of vectors. 

In our work, we first use the LRV algorithm to select $4^4$ renormalized operators from the candidate set made of $4^8$ 
operators $\mathcal{A}(O_{\tau-1})$ of Fig.~\ref{fig:binary}. We then use the `one by one replacement' algorithm to improve this initial choice. 
We now discuss how the choice of renormalized operators impacts the number of repeated measurements of physical observables needed to maintain accuracy.

\subsubsection{Maintaining the accuracy level using renormalized operators}

From now on, let us consider the set of chosen renormalized operators $\lbrace O^i_1 \rbrace = \lbrace {\cal A}(O^i_0)\rbrace$. Since the renormalized operators $O^i_1$ are not orthogonal, it is convenient to construct a set of orthogonal operators, following the approach introduced in~\cite{LP12}. We first define the Gram-matrix $G_{ij} = \text{tr}[O^i_1 (O^j_1)^\dagger]$, and diagonalize it to obtain the matrices $Z$ and $D$ such that $G = ZDZ^\dagger$. Then, we obtain a set of orthogonal operators $\{R_i, i=1,2,3,...\}$, which are eigenvectors of $G$, i.e.,

\begin{equation} \label{ortho}
R^i_1 = \frac{4}{\sqrt{D_{ii}}} \sum_j Z^\dagger_{ij} O^j_1 = \sum_j \beta_{ij} O^j_1
\end{equation}
where we introduced a normalization factor of 4 in order for the operators $R^i_1$ to have the same trace norm as 4-site Pauli observables. Using Eq.~\eqref{ortho}, we can relate the expectation value of the orthogonal operators $R^i_1$ to the expectation values of the physical observables by 

\begin{equation}
 \text{tr}(\rho_1 R^i_1) = \sum_j \beta_{ij} \text{tr}(\rho_1 O^j_1) = \sum_j \beta_{ij} \text{tr}(\rho_0 O^j_0)
\end{equation}
To assess how the number of repeated physical measurements $N_j$ on $O^j_0$ is increased, consider that the measurement of $R^i_1$ is a random variable whose variance is $\mathbb{V}(R^i_1)$. Since physical measurements are performed on different copies of the states, the physical measurements correspond to independent variables and 

\begin{equation}
 \mathbb{V}(R^i_1)=\sum_j |\beta_{ij}|^2 \mathbb{V}(O^j_0) \label{eq:variance}
\end{equation}

Let $M_i(\epsilon)$ be the number of measurements needed to achieve a desired variance $\epsilon$ for $i^{\textrm{th}}$
renormalized observable.
The variance $\mathbb{V}(O^j_0)$ is proportional to the inverse of the number of physical measurements $N_j$ of $O^j_0$. Thus, Eq.~\eqref{eq:variance} becomes
\begin{equation}
 \forall i \quad (M_i(\epsilon))^{-1}=\sum_j |\beta_{ij}|^2 N_j^{-1}
\end{equation}

Let's define the matrix $B_{ij} = |\beta_{ij}|^2$. We want to minimize the total number of measurements 
\begin{equation}
N=\sum_j N_j 
\end{equation}
while maintaining the \emph{minimum} precision $1/M_0$ for any orthogonal operator $R^i_1$.  We thus want minimize $N$ under the condition   
\begin{equation} \label{Restriction}
\forall i  \quad \sum_j B_{ij} N_j^{-1} \leq \frac{1}{M_0} \quad\quad N_j > 0
\end{equation}
Note that we cannot simply choose to minimize the $N_j$ independently since the precision level of different operators are not independent under the condition $N_i > 0$. We found numerically that in most instances we looked at, we cannot avoid the situation in which some observables have better precisions than the others.

Introducing the normalized variables $\tilde{N}_j=N_j/M_0$, we are faced with the optimization problem of minimizing \begin{equation}                                                                                                                              \sum_j \tilde{N}_j                                                                                                        \end{equation}
 under the constraint
\begin{equation} \label{Restriction2}
\forall i  \quad \sum_j B_{ij} \tilde{N}_j^{-1} \leq 1 \quad\quad \tilde{N}_j > 0
\end{equation}

For the $4^4\times4^4$ matrix $B_{ij}$, considering the maximal element $\gamma_j=\max_i B_{ij}$ for every column, we know that 
\begin{equation}
\forall j \quad \sum_i B_{ij} (\gamma_j)^{-1} \leq 4^4                                                                                                                                                                                                                                                                \end{equation}
Thus, a naive choice of $\tilde{N}_j$ would be to choose $4^4 \gamma_j$. Alternatively, we calculated $K_j=\sum_i B_{ij} (\gamma_j)^{-1}$ which is guaranteed to be smaller than $4^4$ and consider the biggest of them $K=\max_j K_j$. We can then take $\tilde{N}_j=K \gamma_j$ to guarantee that Eq.~\eqref{Restriction2} is satisfied. 

The total number of measurements $N$ is 
\begin{equation}
 N=\sum_i \tilde{N}_i M_0
\end{equation}
where $\tilde{N}_i$ can be interpreted as a multiplicative factor which ensures that the estimation of the expectation value using renormalized operators has the same precision as the one obtained using $M_0$ measurements on physical Pauli measurements. To report a single number, we introduce the \emph{conditioning factor} $S$, defined as the average multiplier in the number of measurements 
\begin{equation}
 S\equiv \frac{\sum_i \tilde{N}_i}{4^4} 
\end{equation}

\subsubsection{Estimation of the conditioning factor}

We wrote a simulation code to estimate the conditioning factor $S_{k\to k+1}$ corresponding to the multiplicative factor needed to estimate the 4-site density matrix at level $k+1$ using Pauli measurement at level $k$.

Note that, experimentally, we are interested in the multiplicative factor $S_{0 \to \tau}$ between physical Pauli measurement, i.e. measurement at level 0, and renormalized operator $\{R^i_\tau\}$ at level $\tau>0$, defined by Eq. \eqref{ortho}. Let's consider $\tau=2$ for concreteness. 

We would like to argue that 
\begin{equation}
S_{0 \to 2}\simeq S_{0 \to 1} \times S_{1 \to 2} \label{scale_approx}
\end{equation}
where the approximation comes the fact the ascending superoperator is not distributive, as we now explain.

Since the $\{R^i_1\}$ are orthogonal operators with the same normalization as Pauli operators $\{\Sigma^i_1\}$, there is a unitary transformation mapping between those two set of operators. Thus, mapping $\{R^i_1\}$ or Pauli operators $\{\Sigma^i_1\}$ at level 1 to $\{R^i_2\}$ will have the same overhead because of unitarity. However, when mapping Pauli observables at level 1 to renormalized operators at level 2, we take the tensor product of Pauli operators on two blocks to compute $S_{1 \to 2}$. The renormalized operators $\{R^i_1\}$ on two neighboring block do not obey this tensor product structure. This is illustrated in the bottom figure of Fig.~\ref{fig:explain}.

\begin{figure}
	\includegraphics[width=0.7\columnwidth]{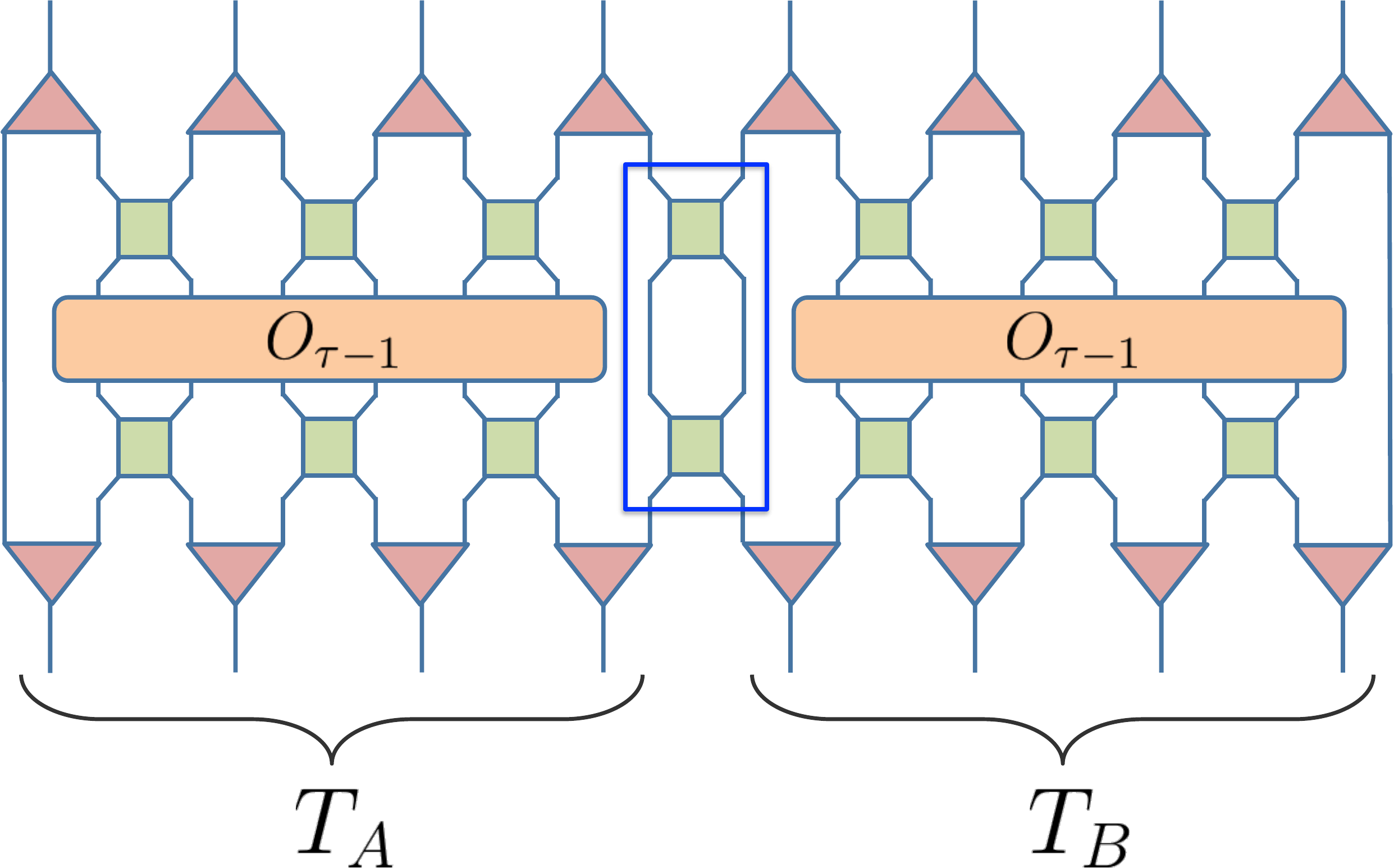}\\
	\vspace{0.1in}
	\includegraphics[width=0.95\columnwidth]{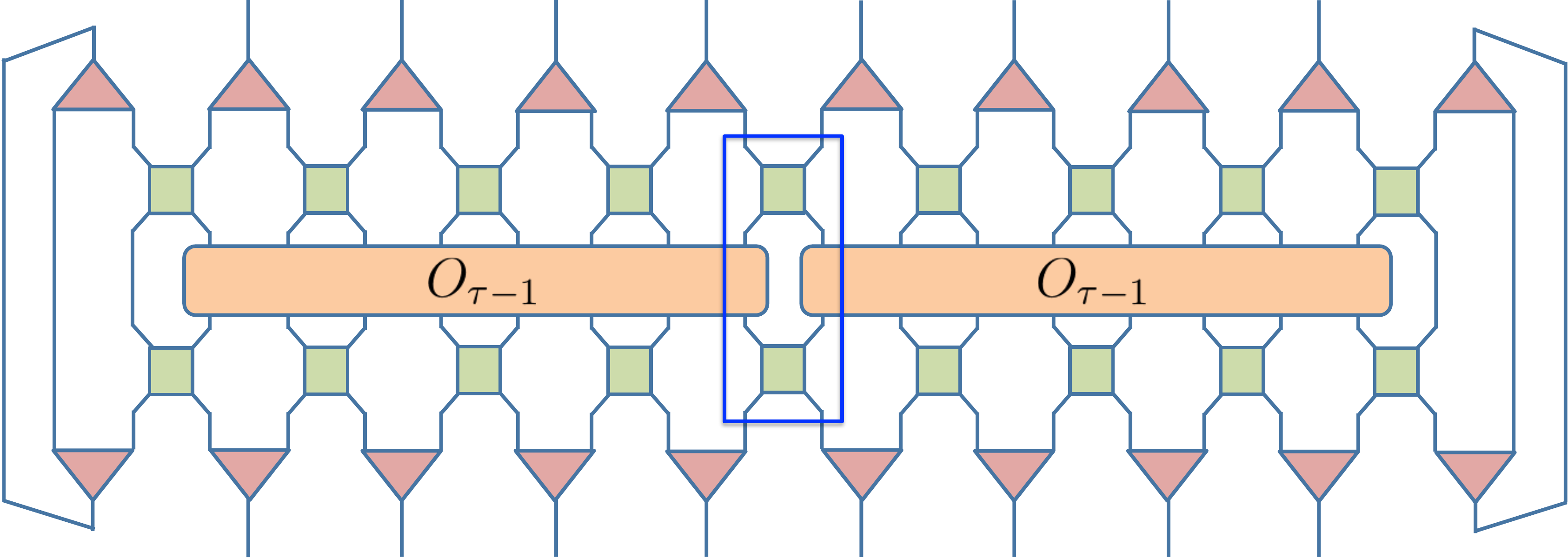}
	\vspace{0.1in}
	\caption{\label{fig:explain} (Color online) The tensor product of $O_{A} \otimes O_{B}$ renormalize under the superoperator $\mathcal{A}$ into the operator resulting from the tensor network contraction. In the upper figure, the distributive law~\eqref{distributive} holds so that the renormalized operators ${\mathcal A}(O_{A})$ and ${\mathcal A}(O_{B})$ can be computed independently and then multiplied. For the figure below, disentanglers within the blue box mixes the renormalized operators and Eq.~\eqref{distributive} does not hold.
	}
	\end{figure}

If the ascending superoperator were distributive, \emph{i.e.,}
\begin{equation}
 {\mathcal A}(O_{A} \otimes O_{B}) = {\mathcal A}(O_{A}) \otimes {\mathcal A}(O_{B}) \label{distributive}
\end{equation}
for physical operators $O_A$ and $O_B$, Eq.~\eqref{scale_approx} would be exact. However, this is not true for  if $O_A$ and $O_B$ are 8-sites Pauli operators and there will be a deviation $\mathcal{E}$ from the distributive law 
\begin{equation}
{\mathcal A}(O_{A} \otimes O_{B}) = {\mathcal A}(O_{A}) \otimes {\mathcal A}(O_{B}) + {\mathcal E}
\end{equation}
resulting from the mixing of operators at the blue box in Fig.~\ref{fig:explain}. However, since $\mathcal{E}$ results from the perturbation of 2 sites, we expect its effect on Eq.~\eqref{scale_approx} to be small since the operators prior to normalization act on 16 sites. We will now see that this intuition is backed by numerical smiulations.

To test the quality of the approximation in Eq.~\eqref{scale_approx}, we performed a simulation to get the exact scaling factor $S_{0 \to 2}$ between the physical level and the second renormalized level and compared it to the product $S_{0 \to 1} \times S_{1 \to 2}$ using a 16-qubit MERA approximation to the groundstate of the critical Ising model. The scaling factor $S_{0 \to 2}$ was obtained by following procedure (see Fig.~\ref{fig:explain}) : for each block $T_A$ and $T_B$, we considered the $4^6$ six-site Pauli operators $O_A^i$ (resp. $O_B^i$) and their renormalized counterparts $\mathcal{A}(O_A^i)$ (resp. $\mathcal{A}(O_B^i)$) to find the optimal basis ($4^4$) maximizing determinant. Then, we have two basis sets with $4^4$ operators $\{\mathcal{A}(O^i_A), i=1,2,...,4^4\}$ and $\{\mathcal{A}(O^i_B), i=1,2,...,4^4\}$. Now, to estimate the reduced density matrix on 4 sites at the second renormalized level, we need to find the best $4^4$ operators out of $\{\mathcal{A}(O_A^i \otimes O_B^j)\}$. The distributive law~\eqref{distributive} holds for $O^i_A$ and $O^j_B$ since they were based on non-interfering six-site operators at the physical level, we can easily calculate the $4^8$ operators $\mathcal{A}(O_A^i \otimes O_B^j)=\mathcal{A}(O_A^i)\otimes \mathcal{A}(O_B^j)$. Now, out of these $4^8$ operators, we renormalize them again using the second renormalized layer, and then find the $4^4$ optimal operators.
                                                                                                                                                                                                                                                                                                                                                                                                                                                                                                                               
We ran the simulation several times and obtained values for the scaling factor $S_{0 \to 2}$ ranging between 23 and 27, which is comparable to $S_{0 \rightarrow 1} \times S_{1 \rightarrow 2}$ which range between 25 and 36. Therefore, we consider the approximation in Eq.~\eqref{scale_approx} to be valid. In fact, the method used to test this assumption gives a scalable way to obtain the optimal set of physical Pauli operators to estimate reduced density matrix at higher renormalized level. 

Now that we assessed the quality of the approximation in Eq.~\eqref{scale_approx}, we will use the formula
\begin{equation}
 S_{0\to\ell} \simeq \prod_{k=0}^{\ell-1} S_{k\to k+1}
\end{equation}
to approximate the total number of measurements needed for MERA tomography at level $\ell$.

We estimated the conditioning factor for MERA tomography on a 24-qubit translation-invariant binary MERA approximation of the ground state of the critical Ising model with periodic boundary condition. Note that the finite system size is too small to reach scale-invariance, which we expect to hold rigorously in the thermodynamic limit. However, translation invariance guarantees that disentangler and isometries are the same in a given layer of the MERA circuit. Since a 24-qubit binary MERA circuit contains 3 renormalization layers, we obtained three conditioning factor $S_{k \to k+1}$ for $k=0,1,2$.

Results for $S_{0 \to \ell}$ for $\ell=1,2,3$ are presented on Fig.~\ref{fig:conditioning} for 10 different MERA approximation of the ground state of the critical Ising model for 24 qubit. Since the disentanglers and isometries are different for every energy minimization, the condition factors also vary. 

\begin{figure}
	\includegraphics[scale=0.6, angle=270]{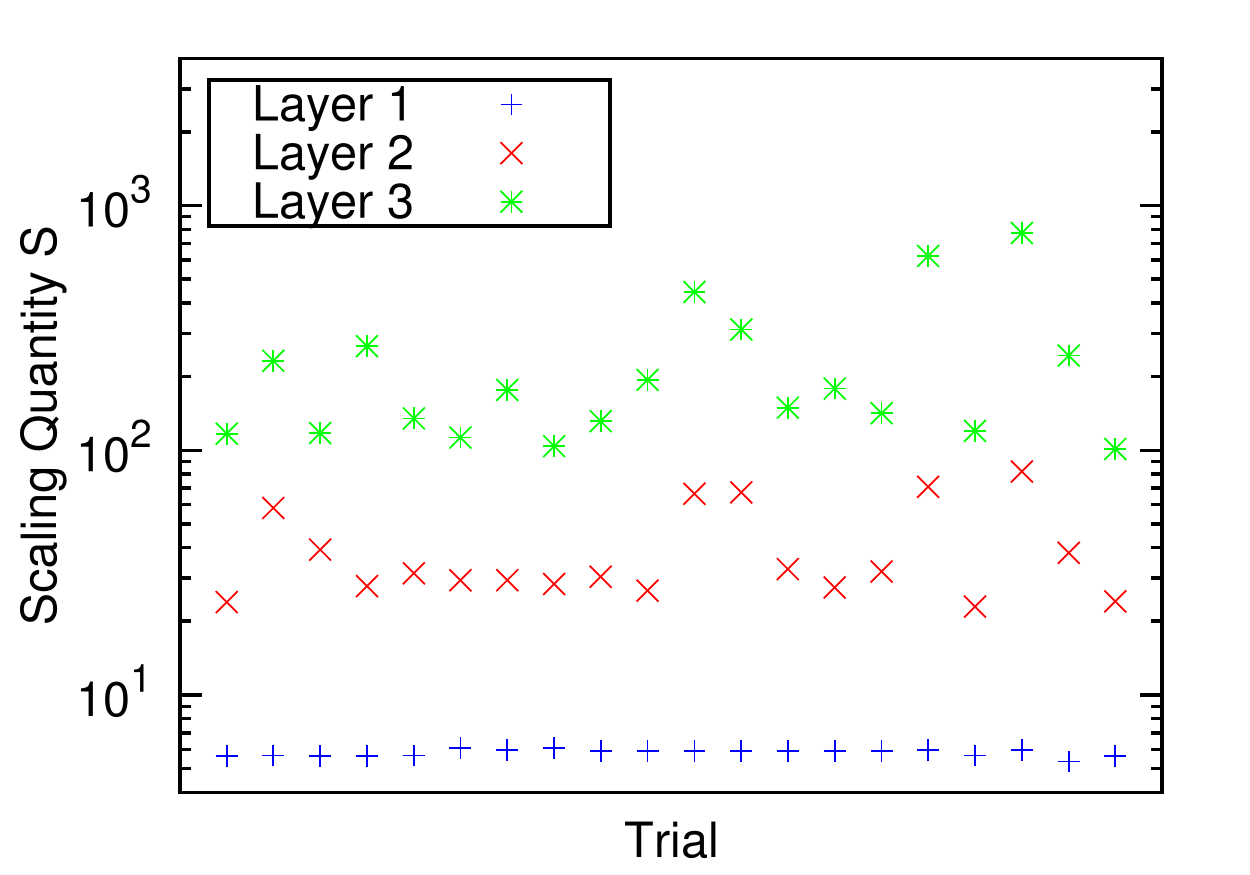}
	\caption{\label{fig:conditioning} (Color online) The behavior of conditioning factor $S_{0\to \ell}$ between level $\ell=1,2,3$ and the physical
	level for different reconstruction of a 24-qubit groundstate of the critical Ising model. We can see that $S_{0\to \ell}$ scales roughly like $6^\ell$. Interestingly, the conditioning factor $S_{0\to 1}$ quantity between the physical level and the first level is very uniform. This is much better than $(\lambda_{block})^k \simeq 2400^\ell$ scaling obtained by using the naive ternary MERA approach.}
\end{figure}

The important feature of the numerical result is that our improvements, in particular the heuristic choice of observables, dramatically improve the scaling of the number of measurements. Indeed, the multiplicative overhead is about $\lambda=6$ between each layers, which is a dramatic improvement over the the multiplicative overhead of $2400$ in the case of the naive ternary approach.

\subsection{Estimates of the total number of measurements required for MERA tomography}

\subsubsection{System size up to 100 qubits}

We are now in position to give an estimate of the total number of physical measurements needed to perform MERA tomography. We estimate these numbers by using the conditioning factor and by choosing the reference number of measurements to be $M_0=100$. This is the number of measurements used to estimate the expectation value of every physical Pauli operators in the tomography of an 8-qubit W state on cold atoms~\cite{HHR+05}.

In Fig.~\ref{fig:actualnumber}, we compare the total number of measurements $N$ for binary and ternary MERA, in both cases for the groundstate of the critical 1D Ising and critical 1D XX models, as a function of the size of a quantum system $n$, i.e., the number of qubits. For ternary MERA, we use the naive approach based on observables which are eigenvectors of the ascending superoperator of Fig.~\ref{fig:scale}. For binary MERA, we used the heuristic choice of observables which maximizes the determinant, obtaining a condition factor $S$ varying between 5 and 6 for critical Ising and between 3 and 3.5 for XX model. 

For the system with total number of qubits ${\cal N} = D\cdot 2^m$ and scaling factor $S$, the total number of measurements was calculated through the formula 
\begin{equation}
 N=100\times \left[ 4^4 \sum_{\tau=0}^{m-3} 2^{m-\tau+1} S^{\tau}   + 4^{D} S^{m-2} \right]
\end{equation}
where $m$ is the total number of layers, $2^{m-\tau+1}$ is the number of isometries between level $\tau$ and $\tau+1$ and the last term comes from the fact that at level $m-2$ there are $D \leq 4$ renormalized particles. The formula was derived in the following way: we assumed that each physical observable was measured with an accuracy of 100 measurements and that this accuracy for maintained for renormalized observables. Thus, for the renormalized observables at layer $\tau$, we need $100\cdot S^\tau$ number of measurements where S is the (average) scaling factor between layers. For each layer, we need to perform brute force tomography on $2^{m-\tau+1}$ number of 4-sites density matrices, each of which having $4^4$ observables. This explains the term inside the summation. The last term arises from the top level of the MERA circuit whose number of sites $D$ is smaller than 4.

On Fig.~\ref{fig:actualnumber}, we also indicated the scaling of brute-force quantum state tomography using the $3^n$ approach of~\cite{AJK04}. The figure confirms the asymptotic advantage of MERA tomography whose polynomial scaling $N \propto n^{\log{S}}$ outperforms the exponential cost of brute-force tomography. Crucially, it also shows this advantage for small system size.

\begin{figure}
	\includegraphics[height=0.92\columnwidth, angle=270]{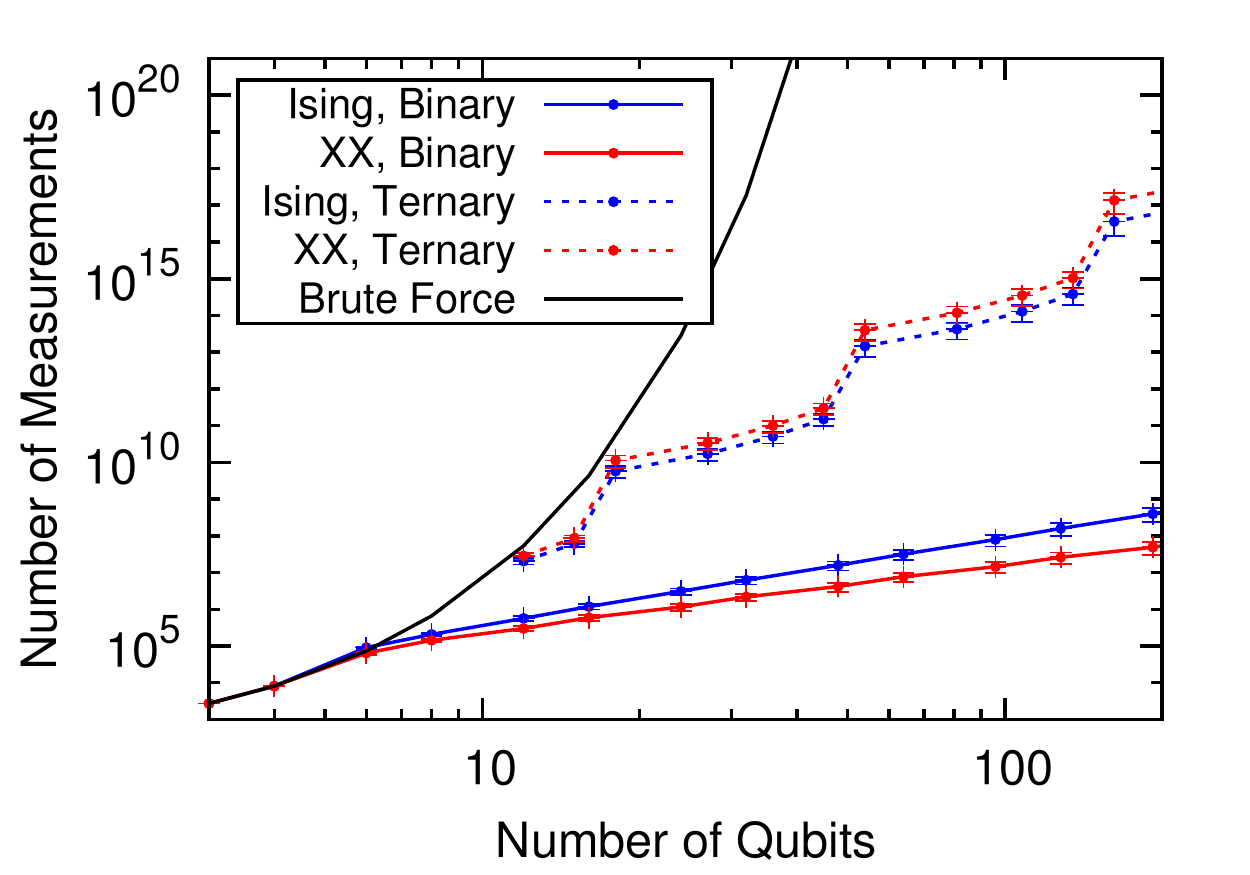}
	\caption{\label{fig:actualnumber} (Color online)
	The number of measurements required versus the number of qubits in MERA tomography.
	For binary MERA, we selected the renormalized operators using the heuristic described in Sec.~\ref{Greedy algorithm} and for ternary MERA, we used the naive approach of taking one-site operators. The error bars account for the uncertainty in the condition factor. We used $5<S_\mathrm{Ising}<6$ and $3<S_\mathrm{XX}<3.5$ 
}
\end{figure}

\subsubsection{Focus on system size up to 24 qubits}

To better appreciate the performance of MERA tomography for system size relevant to experiments, we plotted the total number of measurements needed for binary MERA of the critical Ising and XX models on Fig.~\ref{fig:finite_size}. We compared the number of measurements to the 656,000 measurements used in the largest tomography experiment performed to date, on a 8-qubit system~\cite{HHR+05}.

\begin{figure}
	\includegraphics[height=0.9\columnwidth, angle=270]{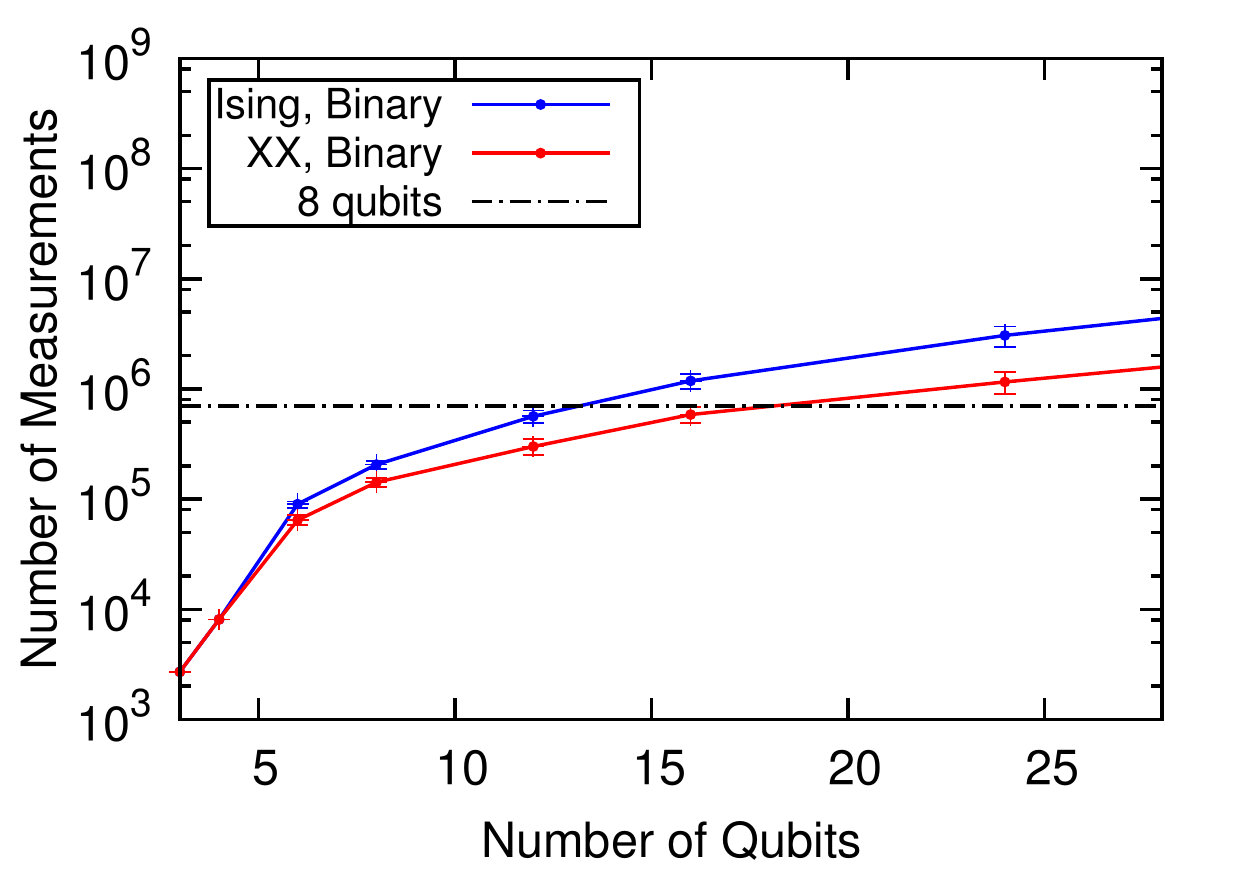}
	\caption{\label{fig:finite_size} (Color online) Magnified version of Fig.~\ref{fig:actualnumber} for binary MERA tomography on the critical Ising and XX models. The black dotted line represents the number of measurements required for qubyte (8 qubits) by brute-force tomography. By optimizing the choice of physical observables, MERA tomography can perform tomography on 16 qubit system (Ising) and 24-qubit system (XX) for the similar number of measurement as brute-force tomography on 8 qubits.}
\end{figure}


Fig.~\ref{fig:finite_size} shows that MERA tomography outperforms brute-force tomography for system sizes that are accessible experimentally, and requires a reasonable experimental effort. More specifically, using our scheme we can perform MERA tomography on a 16-qubit ground state of the critical Ising model and a 24-qubit ground state of the critical XX model with at most twice the number of measurements of the qubyte experiment~\cite{HHR+05} on 8 qubits. Hence, MERA tomography, for a comparable experimental effort, allows to probe quantum systems \emph{twice to three times} larger than brute-force quantum tomography. Moreover, the numerical processing required by MERA tomography is very simple and requires at most a few hours of running time, which is a dramatic improvement over the running time of the Maximum Likelihood Estimation (MLE) used to infer the quantum state compatible with the experimental data~\cite{Blume-Kohout10}.

\section{Propagation of errors} \label{sec:propagation-errors}

In this section, we analyze how errors accumulate and propagate in our MERA tomography scheme. The MERA tomography procedure aims to reconstruct a MERA state $\rho^\mathrm{tomo}$ that approximates the experimental state, defined by
\begin{equation}
\rho^\mathrm{tomo}=U_{0\to m}^\dagger \rho_{m}^\mathrm{trunc} U_{0 \to m}
\end{equation}
where $U_{0 \to m=\prod_{\tau=0}^{m-1} U_{\tau \to \tau+1}}$ is the product of every layer transformation $U_{\tau \to \tau+1}$, i.e, the global MERA circuit and $\rho_{m}^\mathrm{trunc}$ is the output after the final $m$-th layer.
However, our reconstructed state will deviate from the experimental state $\rho_0$ due to (1) imperfect estimation of expectation values of physical observables and (2) truncation errors since each isometry throws out part of the Hilbert space.

A detailed analysis of the impact of truncation errors is presented in the Appendix. We will highlight the key results in this Section and refer the reader to the Appendix for the technical proofs. 

Loss of information in the MERA circuit is due to truncation errors. For every level $\tau \geq 1$, consider $\rho_\tau$ to be the state obtained from the experimental state $\rho_0$ by applying every gates of the quantum circuit before truncation at level $\tau$. Define $\rho_\tau^\mathrm{trunc}$ to be the normalized state obtained by keeping the eigenvectors corresponding to the $\chi$ largest eigenvalues of $\rho_\tau$. Those different states are represented on Fig.~\ref{fig:density2}.

\begin{figure}[h] 
	\includegraphics[width=0.9\columnwidth]{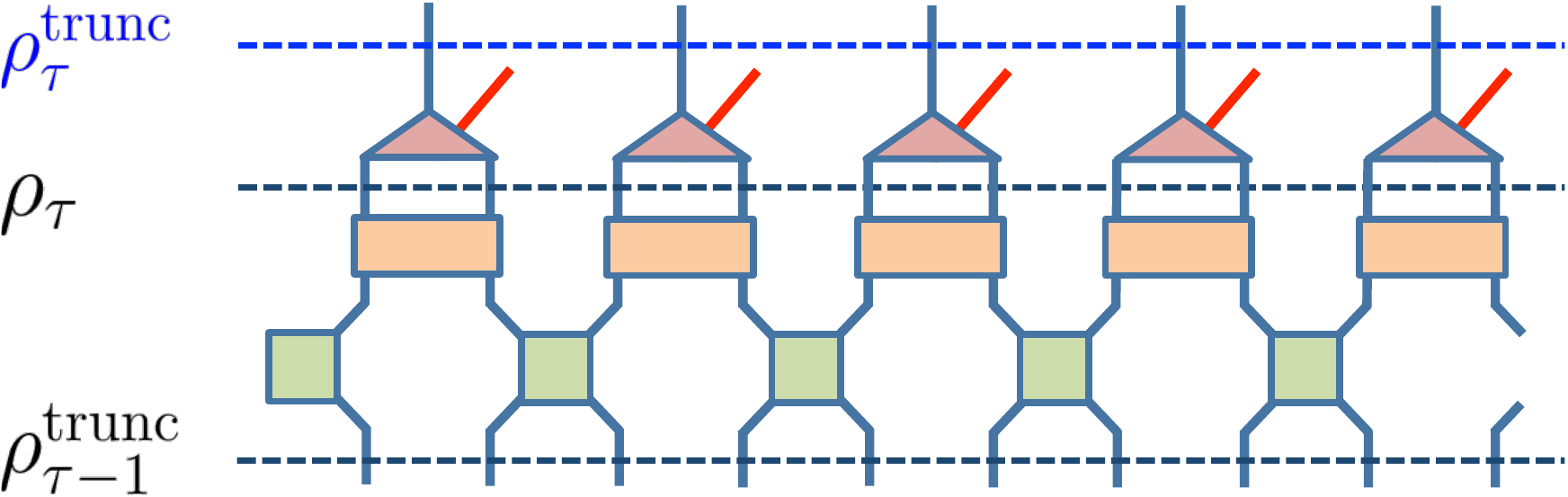}
	\caption{\label{fig:density2} (Color online) A single layer of a MERA circuit which unitarily transform the state $\hat{\rho}_{\tau-1}^{trunc}$ to $\rho_\tau$ before truncating it to $\hat{\rho}_{\tau}^{trunc}$. Red line represents subspace thrown out by isometries.
	}
\end{figure}

We prove in the Appendix (see Lemma~\ref{KeyLemma}) that 
\begin{equation}
 D(\rho_0,\rho^{tomo}) \leq \sum_{\tau=1}^m D(\rho_\tau,\rho_\tau^{trunc}) \label{distance-step1}
\end{equation}
where we use the trace distance $D(\rho,\sigma)=\frac{1}{2}\| \rho - \sigma \|_1$ and $\|A\|_1$ is the sum of the singular values. 

However, the disentanglers and isometries are computed on blocks of the state $\rho_\tau^{rec}$ which is reconstructed for evaluating how expectation values physical observables relate to expectation values of ascended operators on renormalized particles. Because of truncation errors, the ascended operators will be erroneous and the estimated  $\rho_\tau^\mathrm{rec}$ will differ from $\rho_\tau$. We can use the triangle inequality to get
\begin{equation}
 D(\rho_0,\rho^\mathrm{tomo}) \leq \sum_{\tau=1}^m D(\rho_\tau,\rho_\tau^\mathrm{rec}) +  \sum_{\tau=1}^m D(\rho_\tau^\mathrm{rec},\rho_\tau^{trunc}) \label{distance-step2}
\end{equation}

The second term of Eq.~\eqref{distance-step2} is straightforward. It is the intrinsic error introduced by truncation errors and we prove that 
\begin{equation}
 \sum_{\tau=1}^m D(\rho_\tau^\mathrm{rec},\rho_\tau^\mathrm{trunc}) \leq \sum_{\tau=1}^m \sum_k \epsilon^\tau_k \label{term2}
\end{equation}
where $k$ indexes the different isometries at level $\tau$ and $\epsilon^\tau_k$ is the probability weight being removed by the truncation. Note that this term is simply the sum of all truncation errors and will scale \emph{linearly} with the size of the system.

The first term of Eq.~\ref{distance-step2} is due to the relations between renormalized operators and physical observables. It will be related not only to truncation errors at level $\tau$ but also to all truncation errors at previous levels. We prove that 
\begin{equation}
 \sum_{\tau=1}^m D(\rho_\tau,\rho_\tau^\mathrm{rec}) \leq \frac{1}{2} \sum_{\tau=1}^m \sum_{\ell=1}^\tau \sum_k 
 \epsilon^\ell_k \| \sum_{i,j} \beta_{ij} R_i \|_1 \label{term1}
\end{equation}
where the matrix $\beta_{ij}$ and the $R_i$ are defined in Sec.~\ref{sec:Numerical}. Note that this term will scale \emph{quadratically} with system size since truncation errors in previous levels influence truncation errors in subsequent levels. However, since the truncation error $\epsilon$ gets dramatically smaller for higher layer (see Fig.~\ref{fig:numeric}), and we observed numerically that the  summation $\| \sum_{ij} \beta_{ij} R_i \|$ has the order of unity for critical Ising, we expect this contribution to be comparable to the Eq.~\eqref{term2}. 

The final bound will simply be the sum of Eq.~\eqref{term1} and Eq.~\eqref{term2}. Crucially, every term that appears in this bound can be estimated during the tomographic procedure. Thus, it can be used as a \emph{certificate} to check a posteriori if the reconstructed MERA state is indeed close to the experimental state. We expect this to give reasonable bounds in the limit where the truncation errors are very small. Alternatively, one can directly estimate fidelity to assess the closeness between the experimental state and the state obtained by MERA tomography using Monte Carlo fidelity estimation~\cite{FL11,SLP11}. 


\section{Conclusion}
In this work, we investigated and improved upon the original MERA tomography method introduced in~\cite{LP12}.

We showed that the  
scaling of the number of measurements required to maintain accuracy presented in~\cite{LP12} was only valid for single site observable and that its straightforward application to multi-site observables led to an unreasonable overhead.
To circumvent this issue, we suggested to use a different MERA geometry, namely binary MERA, which required performing brute-force tomography on block of 4 renormalized sites (instead of 5 for the ternary MERA case). Furthermore, we introduced a heuristic to identify the physical measurement which give the most information about renormalized particles, in order to minimize the number of physical measurement required. We tested this approach numerically and found that the total number of physical measurements needed to perform MERA tomography on moderate size system is reasonable for experimentalist. For instance, performing MERA tomography on the ground state of the critical Ising model on 16 qubits requires only twice the number of physical measurements needed to perform brute-force quantum state tomography on 8 qubits.
Finally, we gave a deeper understanding of propagation of error in MERA tomography. We bounded the distance between the experimental state and the state reconstructed by MERA tomography in terms of quantities that are estimated locally troughout the tomography procedure. In particular, the propagation of error when using renormalized observables was quantified, and turned out to be closely related to the scaling factor $S$. Since the deviation of the reconstructed state from the experimental state is bounded by a quantity which can be estimated during the tomography procedure, this bound can be used as a certificate to justify \emph{a posteriori} that the experimental state was close to a MERA state.

\begin{acknowledgments}
Most of this work was realized when Jong Yeon was a Summer Undergraduate Research Fellow (SURF) at the California Institute of Technology. OLC acknowledge funding provided by the Institute for Quantum Information and Matter, an NSF Physics Frontiers Center with support of the Gordon and Betty Moore Foundation (Grants No. PHY-0803371 and PHY-1125565) and the Fonds de recherche Qu\'ebec - Nature et Technologies (FRQNT). Both authors would like to thank Glen Evenbly for insightful discussions and numerical help and John Preskill for helpful comments throughout the project. 
\end{acknowledgments}

\bibliographystyle{apsrev}
\bibliography{tomo}

\appendix


\section{Numerical technique to identify disentanglers}
\label{sec:Numerical}

 The conjugate gradient approach used in~\cite{LP12} to find a disentangler and an isometry from a set of measurements data is a standard method for such optimization problem. However, the method is prone to identifying local minima rather than global minima and is hard to extend to large bond dimension $\chi$. Here, we improve the algorithm for optimization by borrowing an algorithm developed for efficient MERA energy minimization procedure in~\cite{EV09}. In this new numerical technique, the objective function remains the same but we interpret it as a tensor contraction
\begin{equation}
f(u,\rho_{i}) = \text{tr}(u\Gamma_{u}^R + u\Gamma_{u}^L )  \label{object2}
\end{equation}
where $\Gamma_{u}$ is the environment of disentangler $u$, shown on Fig.~\ref{fig:TNC}. Since an isometry $W$ keeps only the $\chi$ eigenvectors with largest eigenvalues, contraction of this small part of a MERA circuit gives the sum of first $\chi$ singular values of a reduced density matrix for two sites. 

\begin{figure}[h]
	\includegraphics[width=0.9\columnwidth]{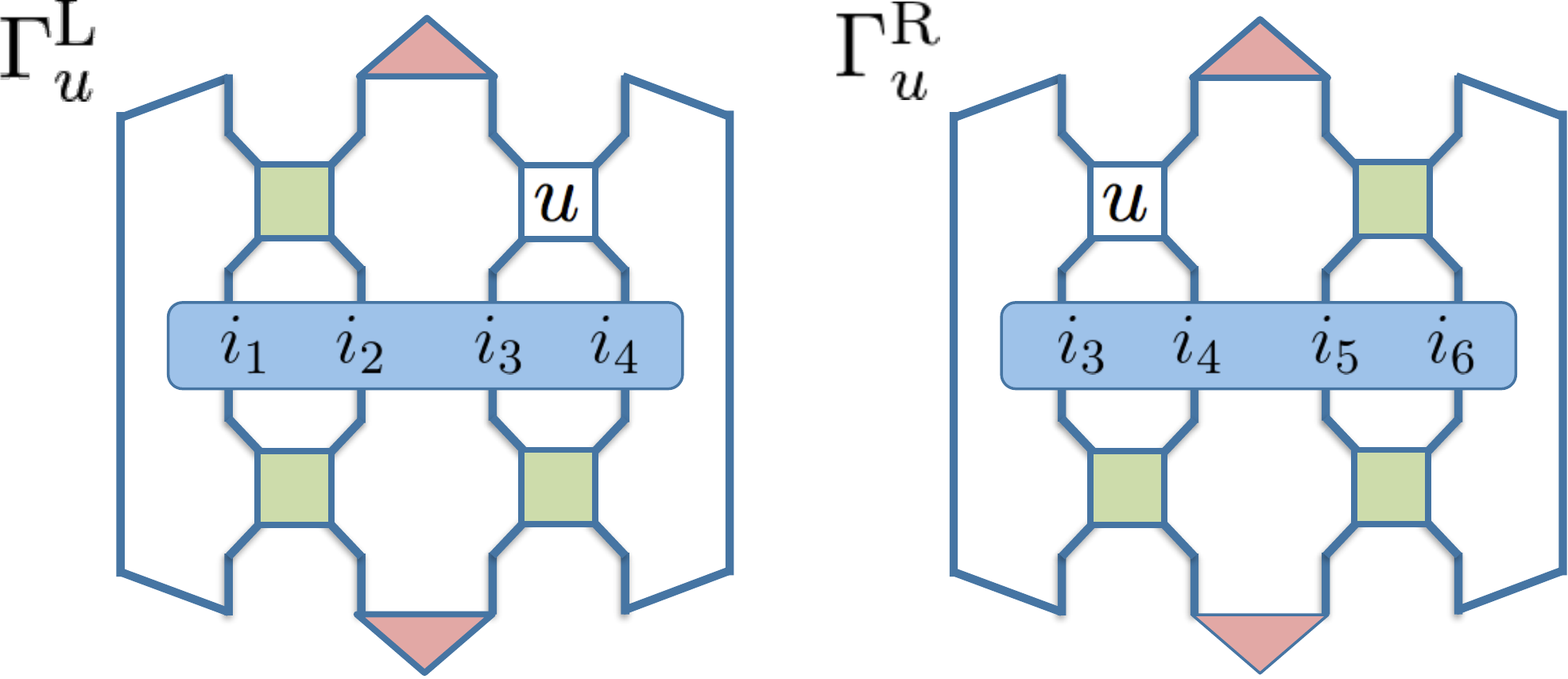}
	\caption{\label{fig:TNC} (Color online) Tensor network contraction corresponding to the objective function $f$ on a ternary MERA. The blue box surrounding the five circles represents the reduced density matrix $\rho_i$ on five sites. The tensor network outside the quantum gate $u$ is the environment $\Gamma_u$ (which includes $u^\dagger$).}
\end{figure}

The optimization of $\text{tr}(u \Gamma_u)$ is analytically hard as $\Gamma_u$ also depends on disentangler through $u^\dagger$. It is a quadratic optimization problem but we will consider its linearized version. We linearize the objective function by fixing $u^\dagger$ and thus $\Gamma_u$ and only varying the disentangler $u$. For a fixed $\Gamma_u$, the optimal $u$ can be found by standard singular value decomposition (SVD) technique. 
One finds the SVD decomposition of the environment, $\Gamma_{u} = NSM^{\dagger}$. The trace in Eq.~\ref{object2} is extremized by the choice of $u = MN^{\dagger}$. Let $u_0$ be an initial guess (random) unitary transformation. Then, starting with $k=1$ and for increasing values $k=2,3,...,$ we obtain $u_k$ from $u_{k-1}$ by optimizing $\text{tr}(u_k \Gamma_{u_{k-1}})$.

Optimizing a disentangler also depends on the neighboring disentanglers. Thus, the optimization not only iterates the above process to optimize the disentanglers, but also sweeps over all the disentanglers of a given layer. An optimization algorithm can choose to balance the number of iterations and the number of sweeps in different ways. Numerically, we find that a single iteration and multiple sweeps gave satisfactory results. Although the objective function is not guaranteed to improve at each step, nor to converge, we find that this method typically converges faster than the previous method based on a conjugate gradient technique.

\begin{figure}
	\includegraphics[scale=0.55, angle=270]{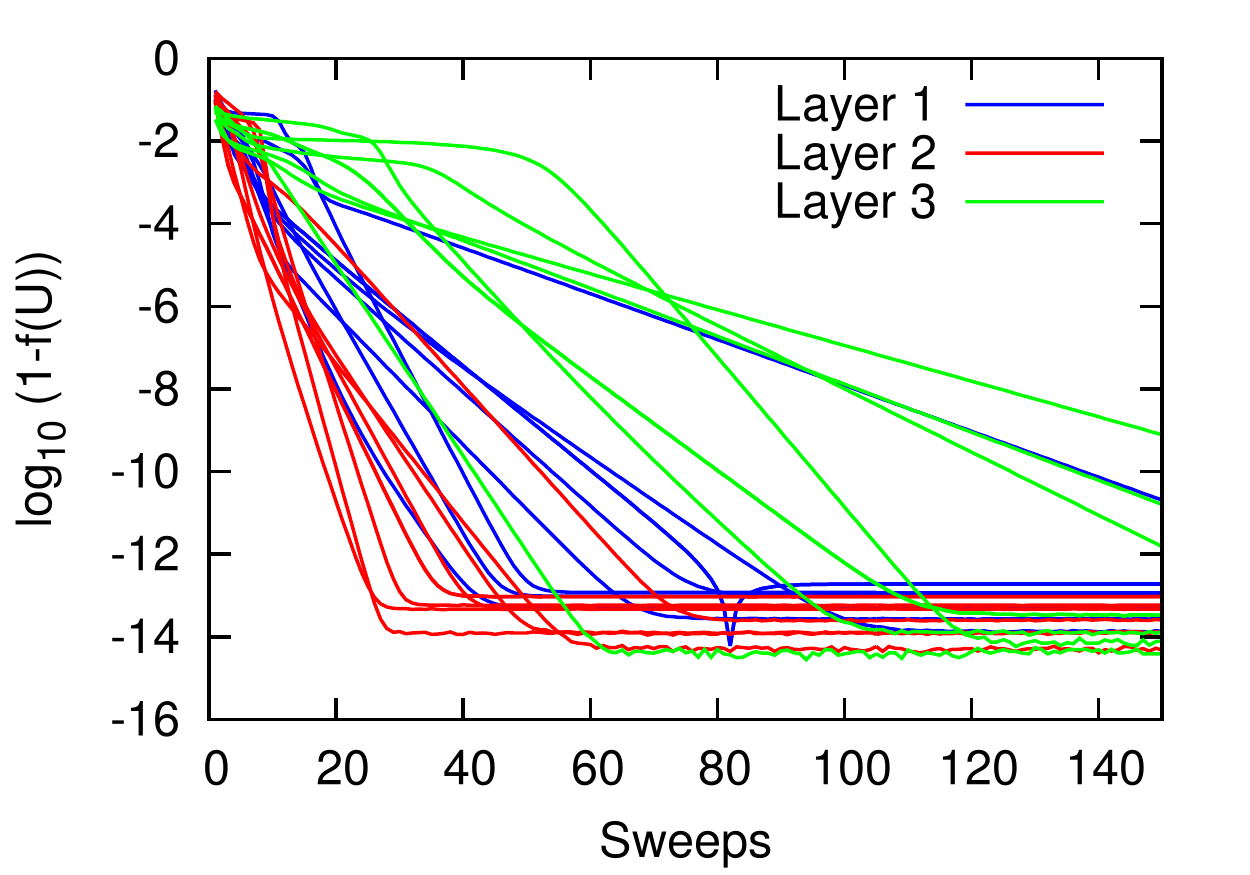} 	
	\includegraphics[scale=0.55, angle = 270]{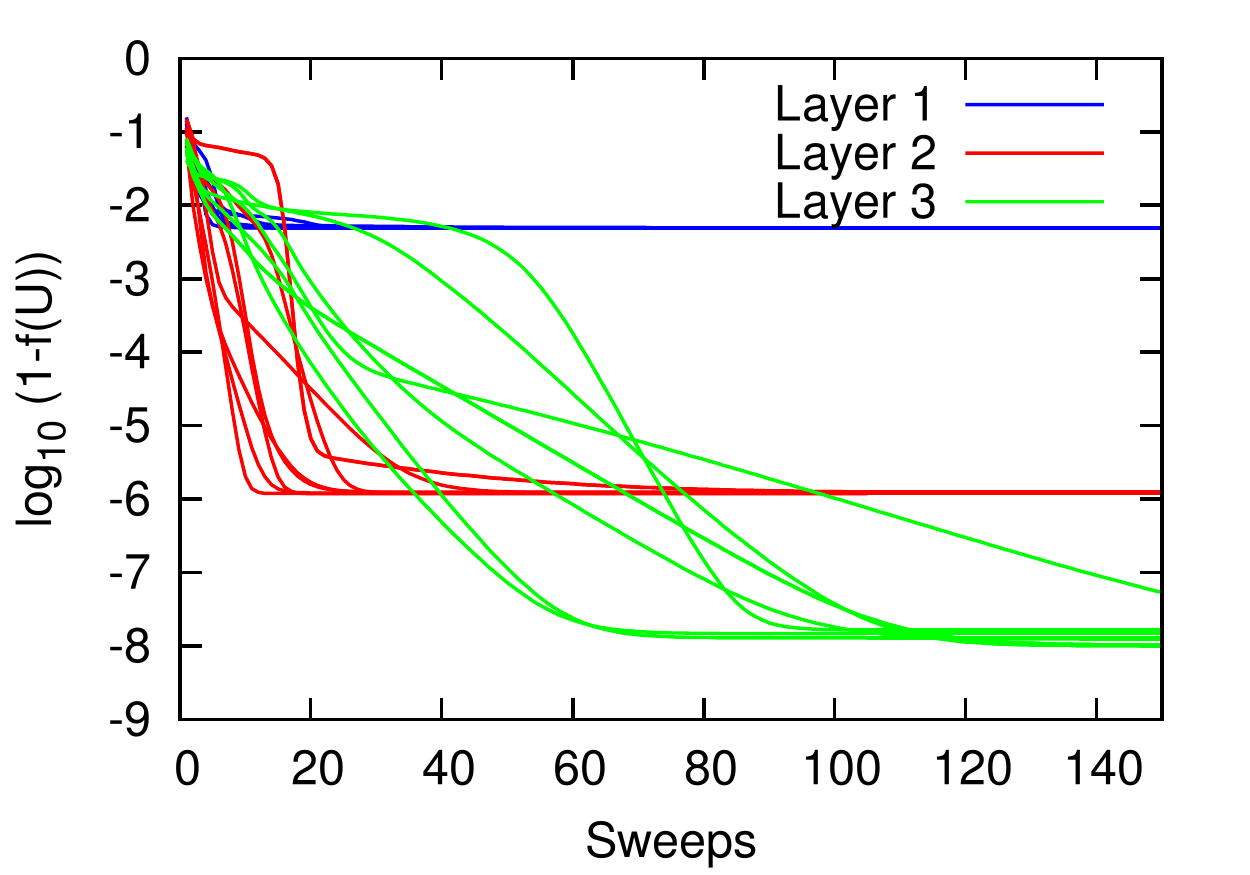} 
	\caption{\label{fig:numeric} (Color online) Convergence of the average value of the objective function $f(u)$ for different layers using MERA tomography on 24 qubits and a binary geometry, as a function of the number of sweeps. (up) Results for 10 random MERA states. (down) Results for 10 superpositions of a MERA state and a Haar-random state of magnitude $\delta=0.1$ (see Eq.~\ref{eq:erroneous_state}). 
}
\end{figure}

To check the validity of our algorithm, we generated random MERA states by choosing the quantum gates of the MERA circuit at random according to the Haar measure. We then performed our quantum tomography algorithm to find MERA circuits for the states. The algorithm was tested on the binary MERA geometry with 24 qubits. In order to check the numerical optimization algorithm, we assumed brute-force tomography to be perfect, implying that the 4-sites reduced density matrices $\rho_i$ are accurate. At the end, the fidelity 
\begin{equation}
F(\rho_0,\rho_r)\equiv \left( \text{tr}\sqrt{\sqrt{\rho_0} \rho^{\mathrm{tomo}} \sqrt{\rho_0}} \right)^2 
\end{equation}
 between the original state $\rho_0$ and the reconstructed state $\rho^{\mathrm{tomo}}$ was computed. The infidelity $1-F(\rho_0,\rho^{\mathrm{tomo}})$ obtained on average was $10^{-13}$. There were few cases which took more than hundred seconds to reach that fidelity, but it was only one out of twenty. In Fig.~\ref{fig:numeric}, we can see that 20 runs all gave convergence in 100 iterations, which take only ten seconds. We conclude that, for states with an exact MERA representation, our method requires reasonable processing time and yields a very accurate reconstruction result. 

Also, we examined how our algorithm performed if the experimental state does not admit an exact MERA representation. 
For this, we considered the experimental state $\ket \psi $ to the be the superposition of a MERA state $\ket{\psi_{MERA}}$ and a Haar-random state $\ket{\psi_e}$
\begin{equation} \label{eq:erroneous_state}
 \ket{\psi} = \sqrt{1-\delta^2} \ket{\psi_\mathrm{MERA}} +\delta \ket{\psi_e}
\end{equation}
Results for that case are given in the bottom plot of Fig.~\ref{fig:numeric}. In that case, the optimization performs poorly on the first layer, converging around $10^{-2}$ and improves for the second layer (around $10^{-6}$) and the third (around $10^{-8}$). The infidelity between the input state and the reconstructed state was around $10^{-2} \approx \delta^2$, which is consistent with the intuition that the reconstructed state is the MERA part of the state. Our interpretation is that the isometries of the layers progressively filter out the non-MERA part of the state.   


\section{Error analysis} \label{sec:ErrorAnalysis}

We want to assess how the errors accumulate throughout the tomography procedure. 
One source of errors is the imperfect estimation of expectation values of \emph{physical} observables due to the finite number of repeated measurements, fluctuations in the state preparation and measurement errors. This source of errors is common to all tomography schemes and putting meaningful error bars on brute-force tomography is a complex issue which is an active area of research~\cite{Blume-Kohout12,CR12}. We will not address this issue here, assuming that expectation values of physical observables are perfect. However, introducing those errors would be straightforward in our analysis. 

Instead, we will focus on the errors which are introduced by the idea of MERA tomography itself. More precisely, the truncation errors at the level of every isometry will introduce i) intrinsic errors since part of the information on the state is thrown away and ii) reconstruction errors since the relationship between renormalized operators and physical observables is not exact.

We will first focus on intrinsic errors by assuming that any reduced density matrix in the circuit can be obtained exactly in Sec.\ref{intrinsic-errors}. In a second step, we will estimate the error introduced by using renormalized operators to extract information about density matrices in higher level of the MERA circuit in Sec.\ref{error-renormalization}.

\subsection{Intrinsic errors} \label{intrinsic-errors}

We first analyze the error propagation at the level of a single isometry in Sec.\ref{sec:error-isometry} in order to infer the propagation of errors for a single layer in Sec.\ref{sec:error-layer} and finally analyze the error propagation for the global MERA circuit in Sec.\ref{sec:error-global}.

\subsubsection{Error for a single isometry}\label{sec:error-isometry}

The experimental state not being a MERA state or imperfection in the numerical optimization of the objective function given by Eq.~\ref{object1} results in the probability weight of the density matrix not being fully in a $\chi$-dimensional subspace (remember that we are discussing the case $\chi=2$ in all numerical examples). 
   
Let $\mathbb{V}$ be the Hilbert space for a particle of quantum dimension $\chi$, i.e., $\mathbb{V}=\mathbb{C}^\chi$. The isometry $w$ maps $\mathbb{V}^{\otimes k}$ to $\mathbb{V}^{\otimes 1}$ where $k=2$ for a binary MERA. It can be conveniently expressed as the product of a unitary transformation $v$ followed by a projector $P$ which maps $\ket{0}^{\otimes {k}}$ to $\mathbb{V}$. In the ideal case, the role of $v$ is to rotate the basis so that a $k$-site reduced density matrix $\rho$ is diagonalized and only has support on the space $\ket{0}^{\otimes {k-1}}\otimes \mathbb{V}$. In practice, $v$ rotates the first $\mu=1\dots\chi$ eigenvectors to $\ket{0}^{\otimes {k-1}}\ket{\mu}$ as represented on Fig.~\ref{fig:density}.

\begin{figure}
	\includegraphics[width=0.9\columnwidth]{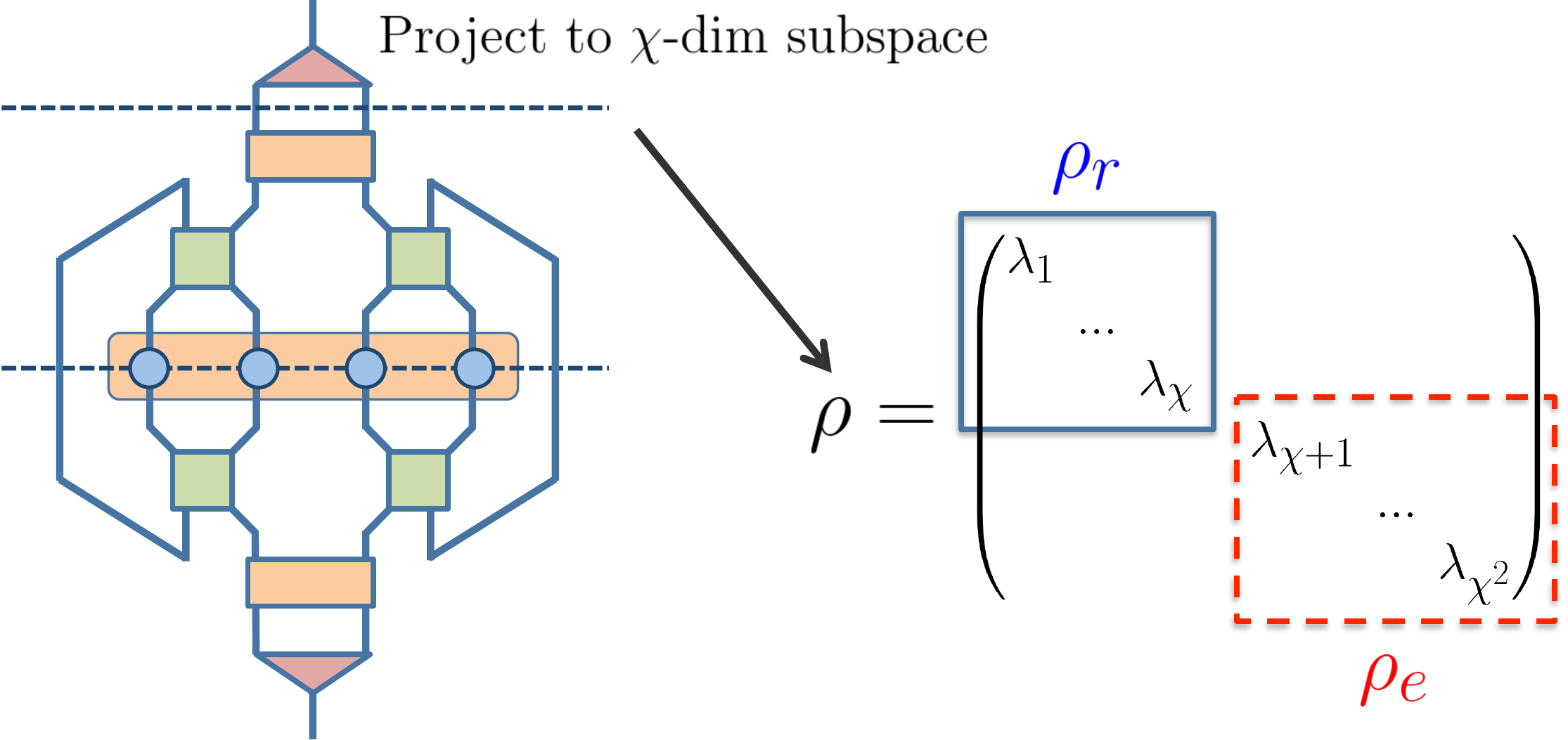}
	\caption{\label{fig:density} (Color online) Tensor contraction for virtual two-sites reduced density matrix $\rho$. $\rho$ represents transformed the density matrix after it passes through the unitary parts of isometries which diagonalizes it. The blue box represents the part of the density matrix which is kept and passed to the next layer of MERA while the red box represents the density matrix that is truncated out.}
\end{figure}

Let $\rho$ be the (virtual) $k$-sites reduced density matrix diagonalized through the isometry, $\rho_r$ be the reduced density matrix supported on $\ket{0}^{\otimes {k-1}}\otimes \mathbb{V}$ and $\rho_e=\rho-\rho_r$ be the density matrix that would be thrown out by the isometry. Under this setting, let $1-f(u) = \sum_{i>\chi} \lambda_i =  \epsilon$. 

The MERA reconstruction procedure to do not lose any information if all entries of $\rho_e$ is zero, i.e., $\epsilon=0$. The renormalization relation for observables, Eq.\ref{eq:scaleO}, was derived under this assumption. However, if the renormalized state has a non-zero $\rho_e$, i.e., 
\begin{equation}
 \rho_1 = u_1 \rho_0 u^\dagger_1 = \rho_r + \rho_e 
\end{equation}
the renormalized observable also as support on the virtual space that is truncated by the MERA circuit
\begin{equation}
O_1 = u_1 O_0 u^\dagger_1 = {\cal A}(O_0) + E
\end{equation}

For convenience, let's consider the normalized states $\rho^\mathrm{trunc}\equiv \rho_r/(1-\epsilon)$ and $\hat{\rho}_e\equiv \rho_e / \epsilon$. Then, in the imperfect MERA setting, Eq.\ref{eq:scaleO} is modified and reads
\begin{equation}\label{eprop}
\text{tr}(\rho_0 O_0) =  (1-\epsilon) \text{tr}\left( \rho^\mathrm{trunc} {\cal A}(O_0) \right) + \epsilon \text{tr}(\hat{\rho}_e E)
\end{equation}

Because we truncate out $\rho_e$ and measure $\rho^\mathrm{trunc}$ in next layer for actual MERA tomography, there will be a discrepancy $|\epsilon \text{tr}(\hat{\rho}_e E)|$ between the measured quantity $\text{tr}(\rho_0 O_0)$ and the quantity of interest $\text{tr}\left( \rho^\mathrm{trunc} {\cal A}(O_0) \right)$. 


In the rest of the discussion, we consider $\epsilon \text{tr}(\rho_e E)$ as a random error about which only $\epsilon$ is known. 

\subsubsection{Error analysis for a layer}\label{sec:error-layer}

Having understood how an error $\epsilon$ affects a single isometry, we will now assess how the errors for each isometry combine inside a layer. 
Let's consider the global state $\rho$ after it goes through the isometries in Fig.~\ref{fig:density2}. As we discussed above, if the MERA circuit is not perfect, the global state $\rho$ would decompose as
\begin{equation}
\rho = \rho_r + \rho_e 
\end{equation}
Since only $\rho_r$ is passed through the MERA circuit and $\rho_e$ is thrown away, $\rho_e$ is informationally inaccessible in the later step of tomography. However, we want to quantify $\rho_e$ in order to bound the distance between $\rho$ and $\rho_r$.


To achieve our purpose, we will divide $\rho_e$ into smaller parts. We will label the isometry with an index $i=1\dots \ell$. Let $\mathbb{C}_i = \ket{0}^{\otimes {k-1}} \otimes \mathbb{V}$ be the correct $\chi$-dimensional subspace kept by the MERA circuit, and $\mathbb{E}_i = \mathbb{V}^{\otimes k} \setminus \mathbb{C}_i$ be the incorrect subspace, indicated by red lines on Fig.~\ref{fig:density2}. Now let $\rho^i$ be the reduced density matrix before the $i$-th isometry. Then, the density matrix at the input of the $i$-th isometry is

\begin{equation}\label{eq:local-density-matrix}
\rho^i =  \rho^i_r + \rho^i_e = \text{tr}_{\bar{i}}(\rho)  
\end{equation}
where $\text{tr}_{\bar{i}}$ represent the partial trace on all sites except the input of the $i$-th isometry and $\rho^i_r\in\mathbb{C}_i$ and $\rho^i_e\in\mathbb{E}_i$ (cf. Fig.~\ref{fig:density}). 

While the globally correct state $\rho_r$ is locally in the correct subspace $\mathbb{C}_i$ for all $i$, the globally incorrect state $\rho_e$ contains part which are locally correct for some $i$ and locally incorrect for some non-empty set $I\subset \left[ 1;\ell\right]$. We will denote $\rho_{e I }$ the part of $\rho_e$ having support on the subspace $\left( \otimes_{i \in I} \mathbb{E}_i \right) \left( \otimes_{i \notin I} \mathbb{C}_i \right)$. Thus,  
\begin{equation}
\rho_e = \sum_{I} \rho_{eI} 
\end{equation}
Note however, that $\rho_{eI}$ is locally correct for $i\notin I$. The terms of the density matrix at the input of the $i$-th isometry of Eq.~\eqref{eq:local-density-matrix} decomposes into  

\begin{eqnarray}
\rho^i_r =  \text{tr}_{\bar{i}}(\rho_r + \sum_{i \notin I} \rho_{eI}) \\
\rho^i_e =  \text{tr}_{\bar{i}}( \sum_{i \in I} \rho_{eI})
\end{eqnarray}

The density matrices $\rho^i$ are estimated by physical measurements and the numerical optimization gives the trace of the locally incorrect state $\text{tr}(\rho^i_e) = 1-f(U) = \epsilon_i$. Knowing this error for each isometry on the layer allow us to estimate the weight of the globally incorrect state

\begin{eqnarray}
\epsilon_e \equiv \text{tr}(\rho_e) & = & \sum_I \text{tr}(\rho_{eI} ) \\
& \leq & \sum_{i=1}^\ell \sum_{i\in I} \text{tr}(\rho_{eI}) \\
& = & \sum_{i=1}^\ell \text{tr}(\sum_{i\in I} \rho_{eI})=\sum_i \epsilon_i \label{eq:error-estimate}
\end{eqnarray}

The normalized truncated state $\rho^\mathrm{trunc} = \rho_r / (1-\epsilon_e)$ is the one we are interested in learning in the next step of variational tomography. The distance between the state before truncation $\rho$ and the normalized truncated state $\rho^\mathrm{trunc}$ can be bounded in fidelity and in trace distance.

For the fidelity, one can observe that 
\begin{equation}
 F(\rho,\rho^\mathrm{trunc}) = \| \sqrt{\rho^\mathrm{trunc}} \sqrt{\rho} \|_1 = \| \sqrt{\rho^\mathrm{trunc}} \sqrt{\rho_r} \|_1=\sqrt{1-\epsilon_e}
\end{equation}
and for the trace distance
\begin{equation}
 \| \rho - \rho^\mathrm{trunc} \|_1 \leq \frac{\epsilon_e}{1-\epsilon_e} \| \rho_r \|_1 + \| \rho_e \|_1 \leq 2 \epsilon_e
\end{equation}

Using those relations and Eq.~\ref{eq:error-estimate}, we get
\begin{eqnarray}
1 - F(\rho, \rho^\mathrm{trunc}) & \leq & \frac{1}{2}\sum_i \epsilon_i \\ \label{inequality}
\frac{1}{2} \| \rho - \rho^\mathrm{trunc} \|_1 & \leq & \sum_i \epsilon_i \label{inequality2}
\end{eqnarray}

\subsubsection{Error analysis of the global circuit}\label{sec:error-global}

The final goal is to estimate the distance between the physical state $\rho_0$ and the reconstructed MERA state 
\begin{equation}
\rho_{tomo}=U_{0\to m}^\dagger \rho_{m}^\mathrm{trunc} U_{0 \to m}
\end{equation}
where $U_{0 \to m}=\prod_{j=0}^{m-1} U_{j \to j+1}$ is the global MERA circuit and $\rho_{m}^\mathrm{trunc}$ is the output after the final $m$-th layer. The idea is to relate this distance to the truncation error $d(\rho^j,\rho^{j}_{rec})$ introduced after each layer $j$ of the MERA circuit when the state $\rho_j$ is truncated to $\rho^\mathrm{trunc}_{j}$. We will use the following lemma.

\begin{lemma} \label{KeyLemma}
For any distance $d(\sigma,\rho)$ which obeys the property $d(U\sigma U\dagger,\rho)=d(\sigma,U^\dagger \rho U^\dagger)$ and the triangle inequality, the following inequality holds 
\begin{equation}
 d(\rho_0,\rho^\mathrm{tomo}) \leq \sum_{\tau=1}^m d(\rho_\tau,\rho^\mathrm{trunc}_j) 
\end{equation} 
\end{lemma}

\begin{proof}
 For m=1, we have 
 \begin{eqnarray}
 d(\rho_0,U_{0 \to 1}^\dagger \rho_{1}^\mathrm{rec} U_{0 \to 1}) & = & d(U_{0 \to 1} \rho_0 U_{0 \to 1}^\dagger, \rho^{1}_{rec})                   \\
 & = & d( \rho^1, \rho_{1}^\mathrm{rec})
 \end{eqnarray}
 
 For arbitrary $m>1$, we have 
\begin{eqnarray}
 & & d(\rho_0,U_{0 \to m}^\dagger \rho_{m}^\mathrm{trunc} U_{0 \to m})  \leq  
  d(\rho_0,U_{0 \to m-1}^\dagger \rho_{m-1}^\mathrm{trunc} U_{0 \to m-1}) \nonumber \\
  & & + d(U_{0 \to m-1}^\dagger \rho_{m-1}^\mathrm{trunc} U_{0 \to m-1},U_{0 \to m}^\dagger \rho_{m}^\mathrm{trunc} U_{0 \to m})    
 \end{eqnarray}
 The last term can be rewritten as
 \begin{equation}
 d(U_{m-1 \to m} \rho_{m-1}^\mathrm{trunc} U_{m-1 \to m}^\dagger,\rho_{m}^\mathrm{trunc})
 =d( \rho_{m},\rho_{m}^\mathrm{trunc})
 \end{equation}
 Recursively applying this inequality proves the lemma.
\end{proof}

We will now apply this lemma to the distance corresponding to the fidelity and the trace distance.

The fidelity can be used to define the distance
\begin{equation}
\theta(\rho,\sigma) =\arccos F(\rho,\sigma)
\end{equation} 
For every layer, we assume that the error is small enough for the approximation $\cos(\theta)\approx 1-\theta^2/2$ to hold. Combined with Eq.~\ref{inequality}, we get
\begin{equation}
 \forall j \quad \theta(\rho_{j},\rho_{j}^\mathrm{trunc}) \approx \left(\sum_i \epsilon^j_i\right)^{1/2}
\end{equation}
where $\epsilon_i^j$ is the error for the $i$-th isometry in layer $j$.

Applying the lemma, we get

\begin{equation}
\theta(\rho,\rho^\mathrm{tomo}) \leq \sum^{m}_{j=1} \theta(\rho_{j},\rho_{j}^\mathrm{trunc})  \approx \sum_{j=1}^m \left( \sum_i \epsilon^j_i \right)^{1/2} 
\end{equation}

To relate the distance to the fidelity, we use the inequality $\cos(\theta)\geq 1-\theta^2/2$ which implies 

\begin{eqnarray}
1 - F(\rho, \rho^\mathrm{tomo}) & \leq & \frac{1}{2} \left[ \sum_{j=1}^m \left( \sum_i \epsilon^j_i \right)^{1/2} \right]^2   \label{ubound1} \\
& = & \frac{1}{2} \sum_{ij} \epsilon^j_i + \sum_{j<j'} \sqrt{ \sum_{i,i'} \epsilon^j_i \epsilon^{j'}_{i'}} \label{interpretation}
\end{eqnarray} 
where the first term in Eq.~\eqref{interpretation} is the incoherent sum of all truncation errors whereas the second term of Eq.~\eqref{interpretation} are due to coherent interference of errors in different layers and isometries.

We can also directly apply the lemma to the trace distance, applying Eq.~\ref{inequality2} to relate the terms to the truncation errors to obtain:
\begin{equation}
D(\rho,\rho^\mathrm{tomo}) \leq \sum_{j=1}^m \sum_{i} \epsilon^j_i \label{ubound2}
\end{equation}

Fig.~\ref{fig:bound} compares the upper bound obtained by Eq.\eqref{ubound1} to the (in)fidelity between simulated experimental states which are the superposition of a random MERA state with a Haar-random state of magnitude $\epsilon$. 
The results on the figure show that the theoretical upper bound for $1-F(\rho,\rho^\mathrm{trunc})$ is a useful proxy.

\begin{figure}
	\includegraphics[scale=0.64, angle=270]{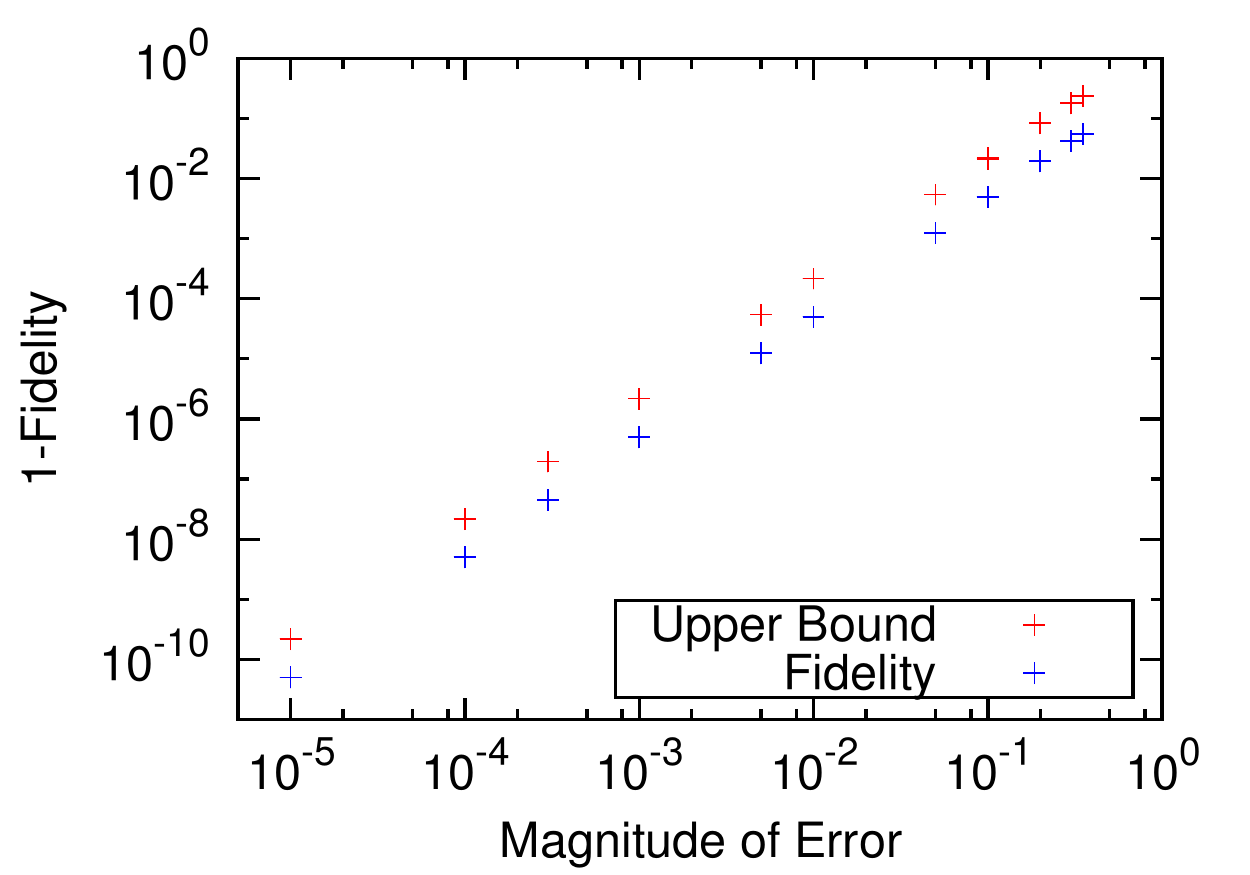} 
	\caption{\label{fig:bound} (Color online) Fidelity between the experimental state and the reconstructed state as a function of the amplitude $\epsilon$ of the Haar-random state added to a 24-qubit random MERA state (see Eq.~\ref{eq:erroneous_state}.
	The upper bound was calculated using values of $\epsilon_i^j$
	obtained through the tomographic procedure. Simulation was performed ten times for each amplitude of error. }
\end{figure}

Note that the magnitude of the Haar-random state $\epsilon$ sets the value of the (in)fidelity between the experimental state and the tomographically reconstructed state since the MERA tomography seems to reconstruct the MERA part of the experimental state, leading to $1-F\approx \epsilon^2$, which appears clearly on Fig.~\ref{fig:bound}. 

%

The upper bounds of Eqs.~\eqref{ubound1} and \eqref{ubound2} can be estimated directly from tomographic data obtained during the reconstruction. Thus, they are a certificate on the distance between the experimental state and the one reconstructed by MERA tomography. 


The error analysis until now assumed that we had access to perfect tomographic estimate of the reduced density matrices on small blocks of particles. However, when we use the structure of the MERA circuit to relate physical measurements to renormalized observables, the truncation error will inevitably introduce errors on the tomographic estimates. We now discuss those type of errors and see how they modify our error bounds.

\subsection{Error introduced by renormalizing physical measurements}\label{error-renormalization}


As described in Sec~\ref{sec:error-isometry}, the truncation errors will not only introduce an intrinsic error, but also lead to an erroneous reconstruction of the reduced density matrix in renormalized layer. Indeed, Eq.~\ref{eprop} which relates the expectation value of the physical observable $\text{tr}\left[\rho_0 O^j_0\right]$ to the expectation value of $\mathcal{A}^\tau (O^j_0)$, the renormalized observable at level $\tau$, on the state we want to reconstruct $\rho_\tau^\mathrm{trunc}$ contains a random error term 
\begin{equation}
\Delta^\tau_j=\epsilon^\tau_e \text{tr}\left[\hat{\rho}_e^\tau E^j\right] 
\end{equation}
where $\epsilon^\tau_e$ is bounded thanks to Eq.~\eqref{eq:error-estimate}.
This erroneous reconstruction will introduce an additional term in the error bound~\eqref{ubound2} which we now analyze.


The reduced density matrix at a renormalized level $\rho^\tau$ will be reconstructed using the the orthonormal operators $R_i$, see~Eq.\ref{ortho}, which span the entire Hilbert space for density operator. Due to the erroneous terms $\Delta^j_\tau$, we have 

\begin{eqnarray} \label{final}
\rho_\tau & = & \sum_i \text{tr}[\rho^\tau R_i] R_i = \sum_{i,j} \beta_{ij} \text{tr}[\rho_\tau O^j_\tau] R_i  \\
 & = & \sum_{i,j} \beta_{ij} \left( \text{tr}[\rho_{\tau-1} O^j_{\tau-1}] - \Delta^j_\tau \right) R_i \\
 & = & \sum_{i,j} \beta_{ij} \left( \text{tr}[\rho_0 O^j_0]- \sum_{\ell=1}^\tau \Delta^j_\ell \right) R_i \\
 & = & \sum_{i,j} \beta_{ij} \text{tr}[\rho_0 O^j_0] R_i - \sum_{i,j} \sum_{\ell=1}^\tau \Delta^j_\ell \beta_{ij} R_i
\end{eqnarray}

Thus, the density matrix reconstructed by our method $\rho_\tau^\mathrm{trunc}$ will be
\begin{equation}
 \rho_\tau^{rec} =\sum_{i,j} \beta_{ij} \text{tr}[\rho_0 O^j_0] R_i=\rho_\tau + \sum_{i,j} \sum_{\ell=1}^\tau \Delta^j_\ell \beta_{ij} R_i
\end{equation}
where the last term quantifies the error due to the erroneous reconstruction, leading to the inequality

\begin{eqnarray}
D(\rho_\tau,\rho^\mathrm{trunc}_\tau) & = & \frac{1}{2} \| \sum_{i,j} \sum_{\ell=1}^\tau \Delta^j_\ell \beta_{ij} R_i \|_1 \\
& \leq & \frac{1}{2} \sum_{\ell=1}^\tau \sum_k \epsilon_{\ell,k} \| \sum_{i,j} \beta_{ij}   R_i \|_1 
\end{eqnarray}

Note that this error term depends not only on the truncation errors at level $\tau$, but also depends on all the truncation errors in previous levels. Thus, this term will in general scale \emph{quadratically} with the size of the system. However, it appears to be well-behaved numerically.

\end{document}